\documentclass[lettersize,journal]{IEEEtran}
\usepackage{amsmath,amsfonts}
\usepackage{array}
\usepackage{textcomp}
\usepackage{stfloats}
\usepackage{float}
\usepackage{url}
\usepackage{verbatim}
\usepackage{graphicx}
\usepackage{cite}
\usepackage{xspace}
\usepackage{hyperref}

\usepackage{wrapfig}
\usepackage{amsthm}
\usepackage{multirow, multicol}
\usepackage{enumitem,kantlipsum}

\usepackage{arydshln}
\newcommand{\nexg}{\textrm{NExG}\xspace}
\newcommand{\neuralex}{\textrm{NeuralExplorer}\xspace}
\newcommand{\reachDest}{$\mathcal{RD}$\xspace}
\newcommand{\NNi}[3][t]{N_{\Phi^{-1}}(#2, #3, #1)}

\newcommand{\bD}{\mathbb{D}}
\DeclareMathOperator*{\argminA}{arg\,min}

\newcommand{\num}[1]{\relax\ifmmode \mathbb #1\else $\mathbb #1$\fi}
\newcommand{\reals}{{\num R}}
\newcommand{\nat}{\mathbb{N}}

\usepackage[linesnumbered,boxed,lined,commentsnumbered]{algorithm2e}
\newtheorem{definition}{\textbf{Definition}}
\newtheorem{theorem}{\textbf{Theorem}}
\newtheorem{lemma}{\textbf{Lemma}}
\newtheorem{assumption}{Assumption}

\newtheorem{remark}{Remark}
\usepackage{subfigure}

\newcommand{\bcheckmark}{\textcolor{red}{\checkmark}}
\SetCommentSty{mycommfont}
\usepackage{xcolor}

\newcommand{\epsr}{\varepsilon_{\text{rel}}}
\newcommand{\epsa}{\varepsilon_{\text{abs}}}
\newcommand{\beps}{\boldsymbol{\varepsilon}}

\begin{document}

\title{NExG: Provable and Guided State Space Exploration of Neural Network Control Systems using Sensitivity Approximation}

\author{Manish Goyal*, Miheer Dewaskar*, and Parasara Sridhar Duggirala
\thanks{M. Goyal is with the Computer Science Department at University of North Carolina at Chapel Hill, NC 27516 USA. (e-mail: manishg.hsr@gmail.com). }
\thanks{M. Dewaskar is with the Department of Statistical Science at Duke University, Durham, NC 27708 USA. (e-mail: miheerdew@gmail.com).}
\thanks{P. S. Duggirala is with the Computer Science Department at University of North Carolina at Chapel Hill, NC 27516 USA. (e-mail: psd@cs.unc.edu).}
\thanks{* Authors contributed equally to this work.}}


\markboth{NExG: Provable and Guided State Space Exploration}%
{Shell \MakeLowercase{\textit{et al.}}: A Sample Article Using IEEEtran.cls for IEEE Journals}


\maketitle

\begin{abstract}
We propose a new technique for performing state space exploration of closed loop control systems with neural network feedback controllers.
Our approach involves approximating the sensitivity of the trajectories of the closed loop dynamics.
Using such an approximator and the system simulator, we present a guided state space exploration method that can generate trajectories visiting the neighborhood of a target state at a specified time.
We present a theoretical framework which establishes that our method will produce a sequence of trajectories that will reach a suitable neighborhood of the target state. 
We provide thorough evaluation of our approach on various systems with neural network feedback controllers of different configurations.
We outperform earlier state space exploration techniques and achieve significant improvement in both the quality (explainability) and performance (convergence rate).
Finally, we adopt our algorithm for the falsification of a class of temporal logic specification, assess its performance against a state-of-the-art falsification tool, and show its potential in supplementing existing falsification algorithms.
\end{abstract}


\section{Introduction}
\label{sec:intro}

\IEEEPARstart{D}{esign} and verification of closed loop systems has become an increasingly challenging task.
First, advances in hardware and software have made it easier to integrate sophisticated control algorithms in embedded systems.
Second, control designers now often integrate multiple technologies and satisfy ever increasing behavioral specifications expected from complex Cyber-physical systems (CPS).
Third, the non-linearity in the behaviors of the closed loop systems makes it difficult to predict the outcomes of perturbations in the state or the environment.
Control design for linear systems typically involves techniques such as pole placement and computing Lyapunov functions. However, such analytical method usually do not scale well to hybrid or  non-linear systems encountered in real world applications. Therefore, we have witnessed a surge in neural network based control design in recent times~\cite{sutton2018reinforcement,2016learningcoordination}.
%
%
But despite its desired utility, neural network  controller, due to its  characteristics and behavior, adds to the complexity of the underlying system thus making it more difficult to perform safety analysis. 

In a typical work flow, the control designer designs a control algorithm and generates a few test cases to check if the specification is satisfied or violated.
However, because of the increasing system complexity, these test cases often do not generalize to the system behavior at large.
This is especially difficult if one has to consider all possible inter leavings of the continuous and discrete behaviors encountered by a modern CPS. Due to different sequence of mode changes, two neighboring states can potentially have divergent trajectories, thus extrapolating the behavior from one state to another becomes challenging.
The problem is further exacerbated by sophisticated neural network based control algorithms.
Since such neural network controllers are typically learned from a finite number of samples, a designer needs to perform additional checks for controller's behavior outside the test suite. However, such manual validation is not practically feasible.

In some instances, the specification is mathematically expressed as as a temporal logic formula that is used by an  off-the-shelf falsification tool for automatically generating a trajectory that violates the specification. But such an approach has a few drawbacks. 
Falsification tools are primarily geared towards finding a violating trace for the given specification, not necessarily to help the designer in systematically exploring the state space. Moreover, the search for a counterexample is performed using stochastic optimization and gradient descent methods. The optimization engine generates random trajectories which would not yield much intuition for the designer about two neighboring behaviors.
Second, if the control designer changes the specification during  testing,
the results from the previous runs may no longer be useful.
Third, existing falsification tools require the specification to be provided in a temporal logic such as signal temporal logic or metric temporal logic (STL/MTL). The designer needs to understand these specification languages which despite being useful in the verification phase, may cause hindrance during the design and exploration phases.

\emph{State space exploration} entails systematically generating trajectories to explore desired (or undesired) outcomes of the system. 
For example, for a given safety specification, a designer might like to generate test cases that are close to satisfying or violating the specification. Existing  falsification tools are neither capable of obtaining such executions nor informing the control designer about the additional tests to conduct for validating safety, measuring coverage or exploring new regions.
Although \neuralex\cite{DBLP:conf/atva/0002D20} alleviates some of these concerns by enabling the control designer to perform systematic state space exploration, it suffers from high training time, and it lacks convergence to the true solution as well as a theoretical analysis.

\begin{figure}[!ht]
\centering     
\subfigure[from S-TaLiRo]{\label{fig:bench9-staliro}\includegraphics[width=68mm, height=40mm]{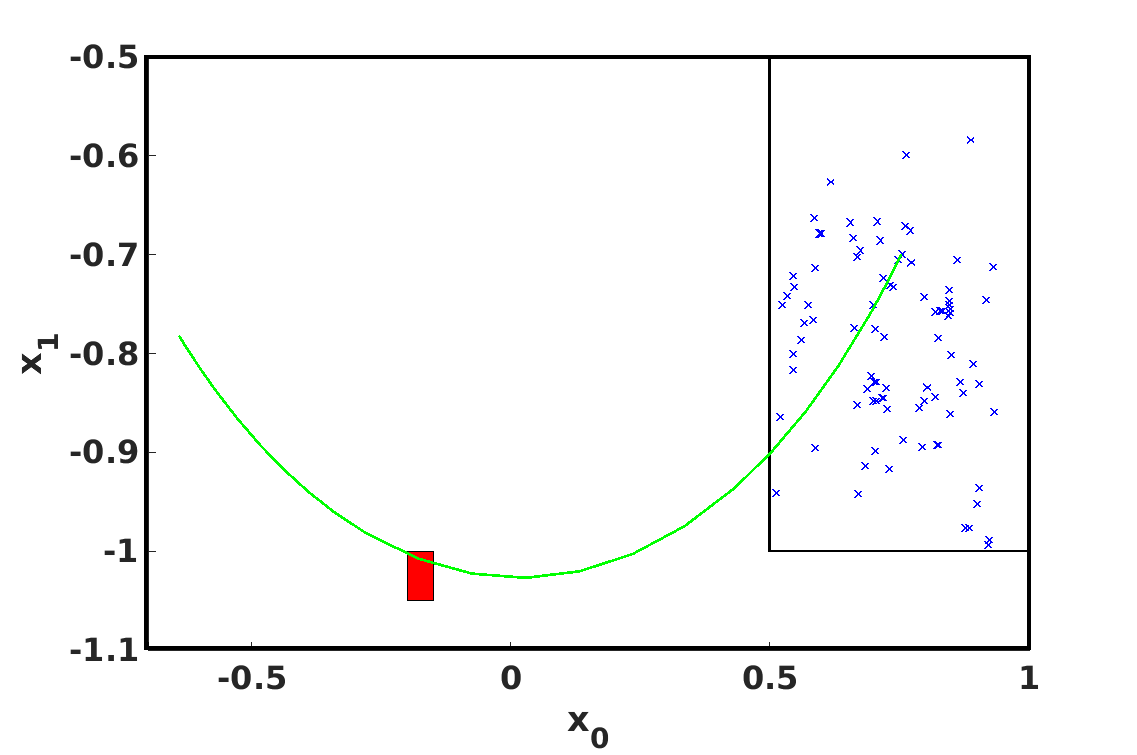}}
\subfigure[from NExG]{\label{fig:bench9-nexg}\includegraphics[width=65mm, height=40mm]{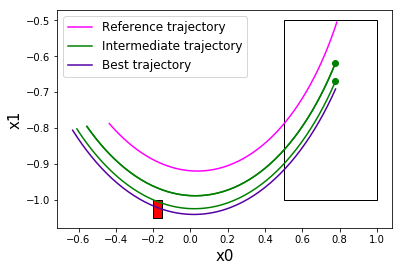}}
\label{fig:intro-falsification}
\caption{\small{Illustration of Falsification schemes. The red-colored box is the unsafe set $U$ and the inner while-colored box is the initial set. Figure~\ref{fig:bench9-staliro} demonstrates that S-TaLiRo conducts falsification in a stochastic manner and takes significant number of iterations to find a falsifying execution to the given safety specification $\neg\Diamond_l U$. Whereas, \nexg finds a valid counterexample in a very few iterations, that too, in a much more directed manner as shown in Figure~\ref{fig:bench9-nexg}.}}
\end{figure}

In this paper, we present a technique that uses neural network approximations of sensitivity over small ranges to to perform state space exploration of closed loop systems with neural network feedback controllers.
The \emph{sensitivity} of a closed loop system at an initial state measures the change in the system trajectory as a result of perturbing the initial state.
Sensitivity can help the designer in developing an intuition about the convergence and divergence of system behaviors. Consequently, the designer can take an active role in systematically navigating the space of new test cases  during control design.
Thus instead of generating executions using a stochastic optimization solver for falsification, our approach explores the state space in a more systematic manner. As a result, it takes  less number of iterations to generate desirable executions as illustrated in Figure~\ref{fig:intro-falsification}.

Since our framework only only requires system traces, it is generalizable to black-box systems in the absence of precise analytical models. Similar to other simulation driven analyses, our technique also depends on using system simulations to tap key information or properties about the system. 
Further, the motivation behind employing neural networks is not only driven by their power to approximate complex functions but also because hardware and software advancements have made these neural networks easy to train and deploy.
We believe that such automated state space exploration is not only useful in man-machine collaborative test case generation, but also for designing safe neural network feedback functions for closed loop control systems.

This paper makes multiple contributions. 
(i) It presents a new state space exploration algorithm \nexg, which is an extension of \neuralex\cite{DBLP:conf/atva/0002D20} for inverse sensitivity learned over small perturbations. 
%
(ii) It provides theoretical guarantees supported by strong empirical results. 
We demonstrate that \nexg will converge to a neighborhood of the target point even if the learned neural network only approximates the (inverse) sensitivity function, and it performs much better than \neuralex while requiring less computational resources. 
The theoretical study also gives insight into why \neuralex does not always achieve convergence in searching for a trajectory that reaches the desired destination. 
(iii) It performs extensive evaluation on 20 standard non-linear benchmarks with up-to 6 dimensions state spaces having neural network feedback controller with multiple layers. 
Empirical evaluations demonstrate that \nexg has better convergence than \neuralex, which is supplemented by significant reduction (up to 65\%) in the training time. 
(iv) It presents a simple inverse sensitivity based falsification algorithm for a class of temporal logic safety specifications. Our evaluations exhibit that the presented falsification scheme is capable of finding a more robust falsifying trajectory in significantly less number of iterations as compared to a widely used falsification tool S-TaLiRo~\cite{AnnpureddyLFS11}.
(v) Finally, it presents additional features of the framework such as computing set coverage, customized  state space exploration, and predicting system trajectories.
%
\section{Related Work}

\emph{Verification} or \emph{Reachability analysis} is typically aimed at verifying safety specification(s) of the safety critical control system~\cite{10.1007/3-540-64358-3_34,10.1007/3-540-36580-X_4}. Some of the notable works in this domain are SpaceEx~\cite{spaceex}, Flow*~\cite{10.1007/978-3-642-39799-8_18}, CORA~\cite{althoff2015introduction} and HyLAA~\cite{bak2017hylaa}. These tools use different symbolic representations such as support functions, generalized star etc. for the set of reachable states. While these techniques are useful for proving that the safety specification is satisfied, some other recent works ~\cite{DBLP:journals/automatica/GoyalD20, DBLP:conf/amcc/0002BD20} have explored using reachability analysis to generate counterexamples of interest.


\emph{Falsification} is employed to generate executions  that violate a given safety specification~\cite{Fainekos09,DonzeM10} instead of proving safety. In these techniques, the required specification is expressed as a formula in temporal logic such as Metric Temporal Logic (MTL)~\cite{MTLmain} or Signal Temporal Logic (STL)~\cite{STLmain,Kyriakis:2019:SMR:3365919.3358231}. For a given temporal logic specification, falsification techniques use various heuristics~\cite{MCFalsification,AbbasF11,SankaranarayananF12,zutshi2014multiple,DeshmukhFKS0J15,ghosh2016diagnosis} in an  attempt to generate trajectories that violate the specification. Two well known falsification tools are  S-TaLiRo~\cite{AnnpureddyLFS11} and Breach~\cite{BreachAD}. Another work~\cite{10.1145/3302504.3311813} uses symbolic reachability supplemented  by trajectory splicing to scale up hybrid system falsification.

Simulation based state space exploration~\cite{donze2007systematic,sensitiveSpaceEx} and verification~\cite{huang2014proofs,duggirala2013verification,fan2015bounded} have also shown some promise by taking the advantages of symbolic and analytical techniques. Such methods either use bounds on sensitivity~\cite{fan2015bounded,huang2014proofs} to obtain an overapproximation of the reachable set or require analytical model to perform random exploration of the
state space~\cite{sensitiveSpaceEx}. While these techniques can bridge the gap between falsification and verification, they might still suffer due to high system dimensionality and complexity. That is, the number of required trajectories may increase exponentially with system dynamics and dimensions. C2E2~\cite{DBLP:conf/hybrid/DuggiralaPM015}, and DryVR~\cite{DBLP:conf/cav/FanQM017} are some of the well known tools in this domain. 

Given the rich history of application of neural networks in control~\cite{miller1995neural,lewis1998neural,moore2012iterative} and the recent advances in software and hardware platforms, neural networks are now being deployed in various control tasks. Consequently, many verification techniques are being developed for neural network based control systems~\cite{ivanov2021verisig,Tran:2019:SVC:3365919.3358230,Sun:2019:FVN:3302504.3311802,DBLP:conf/hybrid/DuttaCJST19,simGuidedReach2021, 10.1145/3302504.3311807} and some other domains~\cite{DBLP:conf/iclr/TjengXT19,DBLP:journals/tecs/SunHKSHA19,DBLP:conf/cav/HuangKWW17}. Many neural network based frameworks for learning the dynamics or their properties have been proposed in recent times~\cite{pmlr-v80-long18a,raissi2018multistep,chen2018neural}, which further underlines the need of an efficient state space exploration.

In the model checking domain, neural networks have been used for state classification~\cite{Phan:2019:NSC:3313149.3313372} as well as reachability analysis by learning state density distribution~\cite{meng_densitydist_2021} or reachability function in \textrm{NeuReach}~\cite{neuReachtacas2022}. In contrast, \nexg learns sensitivity functions and is geared towards state space exploration.
While \neuralex~\cite{DBLP:conf/atva/0002D20} also learns the sensitivity functions, \nexg approximates the sensitivity of closed loop control systems for small perturbations. The corresponding change in the neural network training framework  reduces its training time by considerable amount which is further supplemented with the choice of a uniform network architecture unlike ~\cite{DBLP:conf/atva/0002D20}. 
Two parameters, \emph{scaling factor} and \emph{correction period}, are introduced in \nexg to maintain the trade off between approximation error and the number of simulations generated. 
As a side effect, these parameters render our framework more flexible and amenable to user control.
Another distinguishing feature is that we provide  theoretical guarantees for \nexg (which eluded \neuralex framework) for its convergence.
We discuss these aspects in detail in Sections~\ref{sec:main-algo} and~\ref{sec:convergence}. 

%
%
%
\section{Preliminaries}
\label{sec:prelims}

The state of the system is an element typically denoted as $x \doteq (\mathbf{x}_1, \ldots, \mathbf{x}_n) \in \reals^n$.  For $v \in \reals^{n}$, let $\lVert v \rVert$ denote the standard Euclidean norm of the vector $v$. For $\delta \geq 0$, $B_{\delta}(x) \doteq \{x' \in \reals^n ~|~ \lVert x - x' \rVert \leq \delta\}$ is the closed neighborhood around $x$ of radius $\delta$. We will denote the state space of the system by $\bD \subseteq \reals^n$, and the dynamics of the plant as
\begin{equation*}
    \dot{x} = f(x,u)
\end{equation*}

\noindent{%

where $x \in \bD$ is the state of the system and $u \in \reals^m$ is the input. A closed loop system is a control system where the process or the system is regulated by a feedback control action which is automatically computed as a function of system output.
Suppose we use a feedback-function (also called a controller) $g$ that is regulated by the system output, i.e. $u = g(x)$, then have a \emph{closed loop} system that satisfies 
\begin{equation}
\label{eq:sys}
\dot{x} = f\left(x,g(x)\right).
\end{equation}
We will assume that $f$ and $g$ are such that \eqref{eq:sys} has a unique solution $x: \reals \to \bD$ satisfying $x(0) = x_0$ for every $x_0 \in \bD$. For example, by existence and uniqueness theorem for differential equations, this condition is guaranteed if  $\bD = \reals^n$ and both $f$ and $g$ are  Lipschitz functions of their inputs.}

\begin{definition}

Let $\xi(x_0, \cdot): [0,\infty) \to \bD$ denote the system trajectory starting from the initial point $x_0 \in \bD$. In other words, $x(t) = \xi(x_0, t)$ satisfies \eqref{eq:sys} with $x(0) = x_0$ for $t \geq 0$.
Let $\xi^{-1}(x_1, \cdot): [0, \infty) \to \bD$ denote the backward time system trajectory starting from $x_1 \in \bD$, so that $x(t) = \xi^{-1}(x_1,-t)$ is a solution to \eqref{eq:sys} with $x(0) = x_1$ for $t \leq 0$.
\end{definition}

By the uniqueness of solution to \eqref{eq:sys}, given $x_0, x_1 \in \bD$  and $t>0$ such that $\xi(x_0,t) = x_1$, we have the inverse relation $\xi^{-1}(x_1,t) = x_0$. We now adopt the definitions of sensitivity and inverse sensitivity from~\cite{DBLP:conf/atva/0002D20} as shown in Figure \ref{fig:sensitivity-funcs}.

\begin{definition}
\label{def:sensitivity}
Given an initial state $x_0$, vector $v$, and time $t$, the sensitivity $\Phi(x_0, v, t)$ for the system is defined as
\begin{equation}
\label{eq:sensitivity}
\Phi(x_0,v,t) = \xi(x_0+v,t) - \xi(x_0,t).
\end{equation}

We extend the definition of sensitivity to backward time trajectories, denoted by inverse sensitivity, as
\begin{equation}
\label{eq:invsensitivity}
\Phi^{-1}(x_t,v,t) = \xi^{-1}(x_t+v,t) - \xi^{-1}(x_t,t).
\end{equation}
\end{definition}

\begin{figure}[ht]
\centering     
\includegraphics[width=62mm]{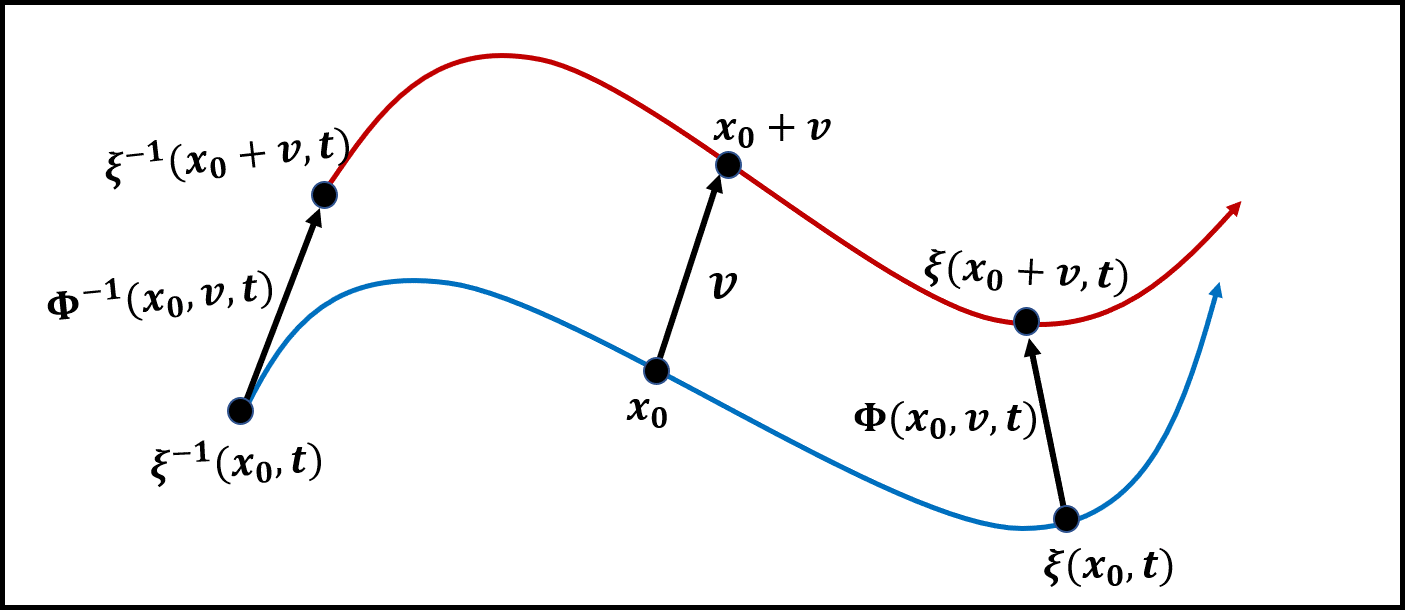}
\caption{\small{Visual description of the sensitivity functions $\Phi$ and $\Phi^{-1}$. The blue and red curves, respectively, denote the unique trajectories that pass through the state of interest $x_0$ and its displaced state $x_0 + v$. The sensitivity $\Phi(x_0, v, t)$ is the displacement between the respective states that the system reaches at time $t > 0$, when  starting from states $x_0 + v$ and $x_0$ at time $t=0$. The inverse-sensitivity $\Phi^{-1}(x_0, v, t)$ is the displacement between the states at $t=0$ that will reach $x_0+v$ and $x_0$, respectively, at time $t > 0$.}}
\label{fig:sensitivity-funcs}
\end{figure}

Informally, sensitivity is the vector difference between states starting from $x_0$ and $x_0 + v$ after time $t$; whereas, inverse sensitivity is the perturbation of the initial state required to displace the state at time $t$ by $v$. In this work, we will primarily focus on using the inverse sensitivity function $\Phi^{-1}$ for performing  systematic state space exploration, but an analogous analysis can also be conducted with the sensitivity function $\Phi$.

For a smooth inverse-sensitivity function $\Phi^{-1}(x_0,v,t)$, let $\nabla_{v} \Phi^{-1}$ denote its Jacobian matrix when considered a function only of its second argument $v$. Then under smoothness assumption, we have the Taylor expansion  
\begin{equation}
\label{eq:approx-sen-func}
    \Phi^{-1}(x_0, v, t) = \nabla_{v}\Phi^{-1}(x_0, 0, t)v + o(\|v\|)
\end{equation}
since $\Phi^{-1}(x_0, 0, t) = 0$.
Therefore learning the inverse-sensitivity function for very small $v$ is akin to learning its directional derivative in the direction $v$. 
\subsection{Learning the inverse sensitivity function using observed trajectories}
\label{sec:approxinvsen}

For testing the system operation on the domain $\mathbb{D}$, one may wish to generate a finite set of trajectories. 
Often, these trajectories are generated using numerical ODE solvers which return system simulations sampled at a regular time step. 
The step size, time bound, and the number of trajectories are specified by the user.
%
%
Given a sampling of a trajectory with step size $h$, i.e., $\xi(x_0,0)$, $\xi(x_0,h)$, $\xi(x_0,2h)$, $\ldots$, $\xi(x_0, kh)$, we make a few  observations. 
First, any prefix of this sequence is also a trajectory of a shorter duration. Hence, from a given set of trajectories, one can truncate them to generate more trajectories having shorter duration.
Second, given two trajectories starting from states $x_0$ and $x'_0$, ($x_0 \neq x'_0$), we can compute the following values for the sensitivity functions:
\begin{align}
\Phi(x_0, x'_0-x_0, t) &= \xi(x'_0,t) - \xi(x_0,t) = x'_t - x_t\label{eq:supplement1} \\
\Phi^{-1}(x_t, x'_t - x_t, t) & = x'_0 - x_0 \label{eq:supplement2}
\end{align}
Note that we can estimate values of $\Phi^{-1}$ based only on samples from a forward simulator $\xi$.

Let us explain how we generate values of the function $\Phi^{-1}(x_t,v,t)$ ($\Phi(x_0,v,t)$) for small values of $v$ in order to learn an approximator $N_{\Phi^{-1}}(x_t,v,t)$ (or $N_{\Phi}(x_0,v,t)$). First, we generate a set of reference trajectories from initial states sampled uniformly at random. Then a fixed number of additional trajectories within a small neighborhood (with radius $\lVert v \rVert \ll 1$) of each initial state are generated. Now, we compute prefixes of the  reference and its neighboring trajectories and use Equations~\ref{eq:supplement1} and~\ref{eq:supplement2} for generating tuples $\langle x_0, v, t, v_{+} \rangle$ and $\langle x_t, v, t, v_{-} \rangle$ such that $v_{+} = \Phi(x_0, v, t)$ and $v_{-} = \Phi^{-1}(x_t, v, t)$. 
We use these tuples to train either a forward sensitivity approximator denoted as $N_{\Phi}$ or an inverse sensitivity approximator $N_{\Phi^{-1}}$. 
Further details on the training procedure for learning the neural networks used in this work are mentioned in Section~\ref{subsec:dir-based-estimator}. The training performance for various benchmark systems and neural network architectures is detailed in Section~\ref{sec:trainingNNs}.

\section{State Space Exploration using local Inverse Sensitivity Approximators}
\label{sec:main-algo}
In this section, we show how to use an inverse sensitivity approximator $N_{\Phi^{-1}}(x_t,v,t)$ for small values of $\|v\|$ in order to perform systematic state space exploration.
State space exploration is typically aimed at finding trajectories that may satisfy or violate a given specification. 
We primarily concern ourselves with safety specifications where the unsafe set is specified as a convex polytopes. In this setup, we would like to find trajectories that reach the set of unsafe states at a specified time, or within a certain time interval.
%
We begin with a sub-routine for state space exploration approach to reach a given destination. We extend this method to a set of states in a subsequent section. 
%

\begin{figure}[!b]
\centering     
\includegraphics[width=80mm]{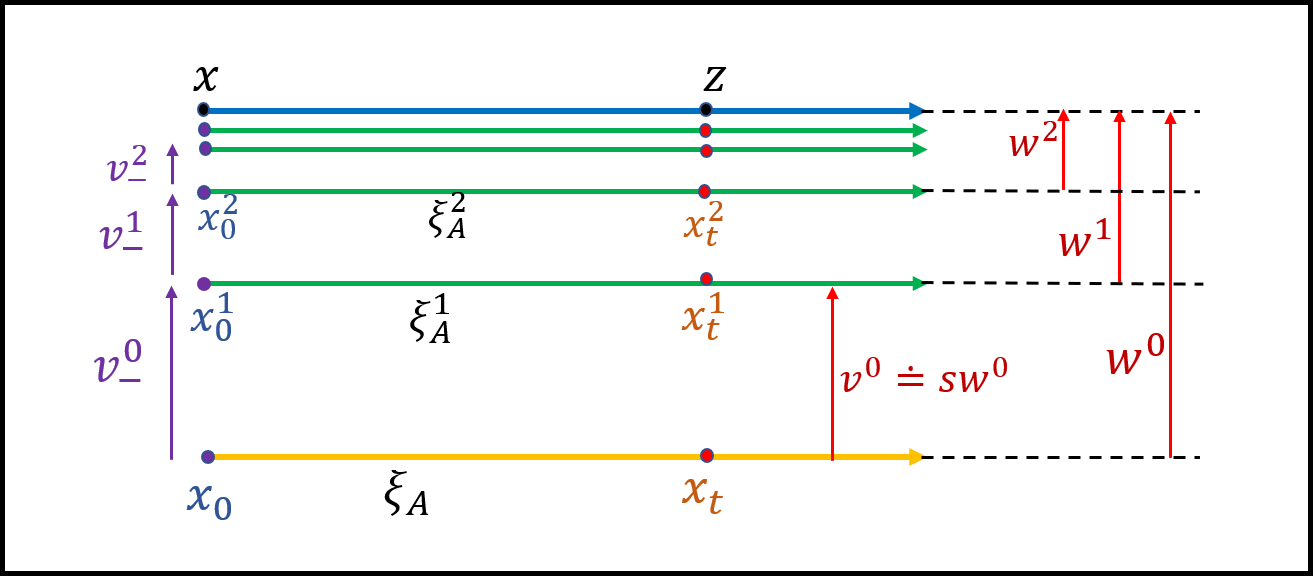}
\caption{\small{Toy execution of Algorithm \ref{alg:reachDest} when our system consists of a constant horizontal vector field in $\reals^2$. Suppose we are given a reference point $x_0 \in \mathbb{R}^2$, the forward simulator $\xi$, the target point $z \in \mathbb{R}^2$ and a time instance $t > 0$. The objective is to find the point $x \in \mathbb{R}^2$, starting from which the system reaches the point $z$ at time $t$ (i.e. $x \doteq \xi^{-1}(z, t)$).  Starting from the initial point $x_0^0 \doteq x_0$, Algorithm \ref{alg:reachDest} successively (for $i=0,1,2\ldots$) increments $x_0^{i}$ by $\nu^{i}_{-} \doteq \Phi^{-1}(x_t^{i}, s (z-x_t^{i}), t)$ for a fixed $s \in (0,1)$, where $x_t^{i} \doteq \xi(x_0^{i}, t)$ is obtained by simulating a new anchor trajectory (denoted by $\xi_A^{i}$) starting at $x_0^{i}$. In this toy example, $\Phi^{-1}$ can be calculated exactly and we have used $p=1$ and $s=0.5$ in Algorithm \ref{alg:reachDest}, but more generally, a local approximator $N_{\Phi^{-1}}$ can be used instead of $\Phi^{-1}$, and the anchor trajectories only need to be calculated once for every $p \geq 1$ steps. The geometric nature of convergence is still preserved under these conditions (Section \ref{sec:convergence}).}}
\label{fig:toy-example}
\end{figure}

\subsection{Reaching a destination at specified time}
\label{subsec:reachDestination}

\begin{figure}[ht!]
\centering     
\includegraphics[width=65mm, height=37mm]{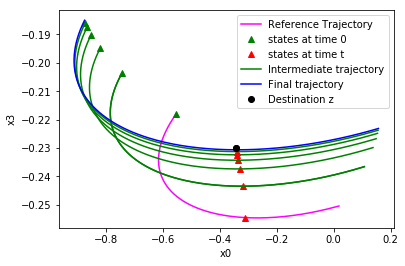}
\caption{\small{Correcting the course of exploration at period $p$ (i.e., simulating a new trajectory after every $p$ steps to aid the search).}}
\label{fig:bench9tanh-cc-4}
\end{figure}

\begin{algorithm}[ht!]
\SetAlgoLined
\SetKwInOut{Input}{input}\SetKwInOut{Output}{output}\SetKw{Return}{return}
\Input{simulator: $\xi$,  time instance: $t \leq T$, reference trajectory: $\xi_{A}$, destination: $z \in \bD$, course corrections bound: $\mathcal{B}$,  function $N_{\Phi^{-1}}$ that approximates $\Phi^{-1}$, initial set: $\theta$, correction period: $p$, scaling factor: $s$, and threshold: $\delta$.}
\Output{course corrections: $k$, final trace: $\xi(x_0^k, \cdot)$, final distance: $d_a^k$, final relative distance: $d_r$}
 $x_0^0, x_t^0 \gets \xi_{A}(0), \xi_{A}(t)$ \tcp*{states at time $0$ and $t$}
 $w^0 \gets z - x_t^0$\tcp*{initial vector difference with $z$}
 $d_{init} \gets d_a^0 \gets \rVert w^0\Vert$ \tcp*{initial distance}
$ k \gets 0$\;
\While{$(d_a^k > \delta)$ \& $(k < \mathcal{B})$}{\label{ln:begLoop1}
  $v^k \gets s \times w^k$\;
  \For{$1 \leq j \leq p$}{\label{ln:begLoop2}
  $\hat{v}_{-}^k \gets N_{\Phi^{-1}}(x_t^k, v^k, t)$ \label{ln:predict} \tcp*{predict $v_{-}^k$}
  $x_0^k \gets \hat{x}_0^{k, \theta} \gets proj_{\theta}(x_0^k + \hat{v}_{-}^k)$ \label{ln:perturbx} \tcp*{perturb $x_0^k$}
  $x_t^k \gets x_t^k + v^k$ \label{ln:perturbx0} \tcp*{progress $x_t^k$}
  }\label{ln:endLoop2}
  $x_0^{k+1} \gets x_0^k$\;
  $\xi_A^{k+1} \gets \xi(x_0^{k+1}, \cdot)$\label{ln:new-anchor} \tcp*{new anchor}
  $x_t^{k+1} \gets \xi_A^{k+1}(t)$\label{ln:coursecorrect} \tcp*{course correction}
 $w^{k+1} \gets z - x_t^{k+1}$ \label{ln:updatev2}\tcp*{new vector difference}
 $d_a^{k+1} \gets \lVert w^{k+1} \rVert$ \tcp*{update distance to $z$}
 $k \gets k + 1$ \tcp*{increment corrections by 1}
 
 }\label{ln:endLoop1}
 $d_r \gets d_a^k/d_{init}$ \tcp*{update relative distance}
 {\bf return} $(k, \xi_A^k, d_a^k, d_r)$\;
 \caption{$\mathcal{RD}$ algorithm aims at finding a trajectory that reaches $\delta$-neighborhood of $z$ at time $t$. 
 It estimates the inverse sensitivity at each step to perturb the initial state, generates a new simulation from perturbed state after every $p$ steps, and treats this simulation as the new anchor. $k$ is the number of simulations generated.
}
 \label{alg:reachDest}
\end{algorithm}

%
In the course of state space exploration, the designer might want to explore the system behavior that reaches a given destination or approaches the boundary condition for safe operation. 
Given a domain of operation, and a sample trajectory $\xi$, the control system designer desires to generate a trajectory that reaches a destination state $z$ (with an error threshold of $\delta$) at time $t$. 
Using previous notation, our goal is to find a state $x$ such that the state $\xi(x, t)$ lies in the $\delta$-neighborhood of $z$.

A toy illustration of our the state space exploration technique with oracle access to the exact inverse sensitivity function is shown in Figure~\ref{fig:toy-example}. 
Given an initial point $x_0$, a destination $z$, and time $t$, we successively move the initial point in small steps in the direction specified by $\Phi^{-1}$, so that the trajectory starting from the new initial point at time $t$ moves closer to that target $z$ with each step. 
In practice, since the exact inverse sensitivity function is unknown, we use a neural-network based approximation instead.

Formally, given an \emph{anchor} trajectory starting from initial state $x_0 \in \theta$ (typically chosen at random), we first compute the  vector $ w^0 \doteq z -x_t^0$ where $x_t^0 \doteq \xi(x_0, t)$.
Next, we estimate the inverse sensitivity $\hat{v}_{-}^0 \doteq N_{\Phi^{-1}}(x_t^0, sw^0,t)$ required at $x_0$ to move $x_t$ towards $z$, and then move $x_0$ by $\hat{v}_{-}^0$. Here, the input $s \in (0,1)$, called as the \emph{scaling factor}, controls the magnitude of movement at each step. This process is again repeated:  move the new initial state $x_0^1 \doteq  x_0 + \hat{v}_{-}^0$ by the vector $\hat{v}_{-}^1 = N_{\Phi^{-1}}(x_t^1, sw^1, t)$, where $w^1 \doteq z-x_t^1$ and $x_t^1 \doteq \xi(x_0^1, t)$ is the point reached at time $t$ by a  new simulated trajectory for the system starting from initial state $x_0^1$. This process is repeated until $x_t^k$ reaches a pre-specified neighborhood of $z$.
%
%

 Since $N_{\Phi^{-1}}$ is only an approximation of $\Phi^{-1}$, the repeated application of the former will typically compound the approximation error. Hence periodically simulating system trajectories starting from intermediate initial states -- a step that we term \emph{course correction} -- is important to keep the exploration on track. Course correction steps not only confirm that the estimates at time $t$ of the trajectory are indeed close to the $z$, but they also allow our procedure to make suitable adjustments if that is not indeed the case.
 
 Since system simulation is expensive, our framework allows for course correction to be performed as frequently as desired. The parameter $p$ is designated as \emph{correction period} because the new anchor trajectory attempts to correct the course once for every $p$ invocations of $N_{\Phi^{-1}}(\cdot)$. Figure~\ref{fig:bench9tanh-cc-4} shows the effect of performing course corrections after every $4$ steps which reduces the number of course corrections to $7$ from $23$ if we corrected the course at every step.
Algorithm \ref{alg:reachDest}, which we call Reach Destination (abbreviated as \reachDest), provides further details of the our procedure. After termination, algorithm \reachDest returns a 4-tuple consisting of: the number of course corrections $k$, the trace $\xi_A^k$ of the last anchor trajectory, the absolute distance $d_a^k$ between the target $z$ and $\xi_A^k$(t), and the relative distance $d_r$. 

%
%

\emph{Notice that the number of course corrections is same as the number of  trajectories $($simulations$)$ generated. If we were to consider physically simulating the plant (which can be expensive) as a part of the operational cost, it would make sense to minimize the number of trajectories we simulate. Limiting the number of simulated trajectories also makes the exploration algorithm more user-friendly by saving time. Thus, we choose the number of course corrections as the primary metric for performance evaluation.}

\section{Theoretical analysis of the convergence of \textsf{ReachDestination}}
\label{sec:convergence}

We now discuss the convergence of Algorithm \ref{alg:reachDest}. 
As seen in Figure \ref{fig:toy-example}, the distance between $x^i_t$ and the target $z$ contracts by a factor of $0 < 1-sp < 1$ in each iteration if the exact inverse sensitivity is used. That is, 

\begin{equation}
\label{eq:convergence}
   \|x^{k}_t - z\| \leq (1-sp)^k \|x^0_t - z\| 
\end{equation}
%
Hence, the generated trajectory will reach the desired destination within an error of $\delta$ after $k^*$ iterations where,
\begin{equation}
\label{eq:k-simple}
k^* = \left\lceil \frac{\log (\|x^0_t - z \|/\delta)}{-\log (1-sp)}\right\rceil.
\end{equation}

However, in the \reachDest algorithm, instead of the exact inverse sensitivity function, we only use its approximation.
In this section, we show that it is possible to achieve a similar geometric rate of convergence even with an approximation.
%
%
Note however that the convergence of \reachDest can fail badly in cases when the system is chaotic or the approximation error is large. To this end we now make assumptions on the regularity of the system and the magnitude of the approximation error that will allow for performance guarantees for \reachDest.

\begin{assumption}
\label{ass:etas}
Suppose there are functions $\eta_1, \eta_2 : [0,T] \to [0,\infty)$ so that
    \begin{equation}
        \label{eq:discr}
        \eta_1(t) \| x - x'\| \leq \|\xi(x,t) - \xi(x', t) \| \leq \eta_2(t) \| x - x'\|
    \end{equation}
    for each $x, x' \in \bD$ and $t \in [0,T]$.
\end{assumption}

 The functions $\eta_1$ and $\eta_2$, sometimes called as witnesses to the discrepancy function \cite{duggirala2013verification}, provide worst-case bounds on how much the distance between trajectories expand or contract starting from different initial states. These functions (and their ratios) can be considered as a measure of the regularity of the system.  Although in practice it may be hard to obtain the  values $\eta_1$ and $\eta_2$ for the system at hand, exponential lower bound for $\eta_1$ and a similar upper bound for $\eta_2$  can be obtained using Gr\"onwall's inequality under a  Lipschitz continuity assumption on the vector field. As shown in the following lemma, Assumption \ref{ass:etas} also ensures that  $\Phi^{-1}(x,v,t)$ is a Lipschitz function of its inputs $x$ and $v$. This is important as such functions can be approximated by Neural networks of bounded depth (see e.g. \cite[Theorem 4.5]{guhring2020expressivity}). 
\begin{lemma}
 If Assumption \ref{ass:etas} is satisfied then for any $t \in [0,T]$ 
 $$
 \|\Phi^{-1}(z',v',t) - \Phi^{-1}(z, v, t)\| \leq \left( 2\|z-z'\| + \|v-v'\|\right)/\eta_1(t)
 $$
\end{lemma} 
\begin{proof} 
By taking $(x,x') = (\xi^{-1}(y,t), \xi^{-1}(y',t))$ in Assumption \ref{ass:etas}, note that $\|\xi^{-1}(y,t)-\xi^{-1}(y',t)\| \leq \|y-y'\|/\eta_1(t)$ for any $y,y'  \in \bD$. The Lemma now follows by suitably applying triangle inequality and the definiton of $\Phi^{-1}$.
\end{proof}
 
In general, we will use the following model to measure the approximation error of $N_{\Phi^{-1}}$. The separate roles played by the relative error $\epsr$ and the absolute error $\epsa$ will become more clear in the context of Theorem \ref{thm:bound-2}. 

\begin{definition}
\label{def:nn-approx-defn}
$N_{\Phi^{-1}}$ is called an $(\epsr,\epsa)$-approximator of  $\Phi^{-1}$ upto radius $r$ and time $T$  if
\begin{equation}
\| N_{\Phi^{-1}}(x_t, v, t) -  \Phi^{-1}(x_t, v, t) \| \leq \epsr \| \Phi^{-1}(x_t, v, t)\| + \epsa \nonumber
\end{equation}
for any $x_t \in \mathbb{D}$, $t \in [0,T]$ and $v \in \reals^n$, with $\|v\| \leq r$. 
\end{definition}

We are now ready state Theorem \ref{thm:bound-2} which bounds the distance between $z$ and iterate $x_t^k$ in the $k$th iterations of the outer loop in \reachDest  when the systejm satisfies Assumption \ref{ass:etas} and $N_{\Phi^{-1}}$ satisfies Definition \ref{def:nn-approx-defn} with sufficiently small error terms  $(\epsr,\epsa)$. To further interpret Theorem \ref{thm:bound-2}, note that: 
\begin{enumerate}[leftmargin=*]
    \item When the additive error $\epsa \approx 0$ is negligible and the relative error satisfies $\epsr \in [0, \eta_1(t)\eta_2(t)^{-1}]$, Equation \ref{eq:thm-bound} holds for any $k \in \nat$ with $r_{\beps}(t)/s \approx 0$. Hence, in this case, a geometric convergence similar to that described for the toy example from above continues to hold with a slightly slower convergence rate (i.e. $-\log(1-sp\gamma_{\beps}(t))$ instead of  $-\log(1-sp)$).
    \item On the other hand, when $\epsa$ is non-negligible (but sufficiently small so that $r_{\beps}(t) \leq r$), the last term in Equation \ref{eq:thm-bound} cannot be ignored. In this case, if rest of the assumptions of Theorem \ref{alg:reachDest} are also satisfied, one obtains the guarantee that $\lim_{k \to \infty} d_a(k) \leq r_{\beps}(t)/s$. Hence if $r_{\beps}(t)/s < \delta$, the termination condition $x_t^k \in B_{\delta}(z)$ will eventually be satisfied whenever $k \geq k^* \doteq \lceil\log(\frac{\delta - r_{\beps}(t)/s}{d_{init}})/\log(1-sp\gamma_{\beps}(t))\rceil$.
    \end{enumerate} 

\begin{theorem}[Convergence of \reachDest]
\label{thm:bound-2}
Fix the domain $\bD=\mathbb{R}^d$ and a time $T > 0$. Suppose 
\begin{enumerate}
   \item The system $\xi$ satisfies Assumption \ref{ass:etas}, 
   \item $N_{\Phi^{-1}}$ is an $(\epsr, \epsa)$-appoximation of $\Phi^{-1}$ for radius $r$ and time $T$, and
   \item $\epsr, \epsa \geq 0$ values small enough so that for each $t \in [0,T]$, $\gamma_{\beps}(t) \doteq 1 - \epsr \eta_2(t)\eta_1(t)^{-1} > 0$ and $r_{\beps}(t) \doteq \epsa\eta_2(t)/\gamma_{\beps}(t) \leq r$.
\end{enumerate}
Suppose the inputs $\theta = \mathbb{D}$, $t \in [0,T]$,  the correction period $p$ and the destination $z \in \mathbb{D}$ to Algorithm $\ref{alg:reachDest}$ are given   
 and the scaling factor satisfies $s \in [r_{\beps}(t)/d_{init}, \min(r/d_{init},1/p)]$, where $d_{init} \doteq \|x_t^0 - z\|$ is the distance between the point $x_t^0 \doteq \xi(x_0^0, t)$ and the destination $z$. Then 
the distance after $k$ iterations of the outer loop in Algorithm~\ref{alg:reachDest} satisfies the following bound 
%
\begin{equation}
    \label{eq:thm-bound}
        \| x^k_t - z \|  \leq (1- sp\gamma_{\beps}(t))^k  d_{init} + \frac{r_{\beps}(t)}{s} 
    \end{equation}
for any $k \in \mathbb{N}$.

\end{theorem}

The proof of Theorem \ref{thm:bound-2} can be seen to be a suitable contraction argument. Formal details are given below.
\begin{proof}[Proof of Theorem \ref{thm:bound-2}]

    In this proof, for mathematical clarity, we slightly change the notation for the variables used in Algorithm \ref{alg:reachDest}. For each $k \geq 0$, let $x_0(k), x_t(k), w(k)$ and $d_a(k)$ denote the values of the variables $x^k_0, x^k_t, w^k$ and $d_a^k$ after $k$ executions of the outer loop in Algorithm \ref{alg:reachDest}. Hence the equalities $w(k) = z - x_t(k)$, $d(k) = \| w(k)\|$, and $x_t(k) = \xi(x_0(k),t)$ are satisfied for any $k \geq 0$.
    
    Since $\theta = \bD$, unwinding the inner loop in Algroithm \ref{alg:reachDest}, note
    \begin{equation}
    \scriptsize
        \label{eq:x0increment}
        x_0(k+1) = x_0(k) + \sum_{l=1}^p \NNi{x_t(k)+(l-1)sw(k)}{sw(k)}. 
    \end{equation}
    Let $\tilde{y}(k) \doteq x_t(k+1) - x_t(k)$ denote the increment between the trajectory end points between the $k$ and $(k+1)$th iteration. Using the fact that $x_t(k) = \xi(x_0(k),t)$, note 
    \begin{equation}
    \scriptsize
        \label{eq:actualw}
        \tilde{y}(k) = \xi\left(x_0(k+1), t\right) - \xi(x_0(k), t).
    \end{equation}
    The quantity $\tilde{y}(k)$ approximates the true target increment given by 
    \begin{equation}
    \scriptsize
    \label{eq:truew}
     y(k) \doteq \xi\left(x_0(k) + \Phi^{-1}(x_t(k), sp w(k), t), t\right) - \xi(x_0(k), t),
    \end{equation}
    which, using Definition \eqref{eq:invsensitivity} of $\Phi^{-1}$, satisfies
    {\scriptsize
    \begin{align}
        y(k) &= \xi\left(x_0(k) + \xi^{-1}(x_t(k) + sp w(k), t) - \xi^{-1}(x_t(k),t) , t\right) - \xi\left(x_0(k), t\right) \nonumber \\
            &= \xi\left(x_0(k) + \xi^{-1}(x_t(k) + sp w(k), t) - x_0(k), t\right) - x_t(k)  \nonumber \\
            &= x_t(k) + sp w(k) - x_t(k) = sp w(k). \label{eq:walternate}
    \end{align}}%
    
\noindent{Subtracting \eqref{eq:truew} from \eqref{eq:actualw}, and using the upper bound from \eqref{eq:discr}}
    {\scriptsize
    \begin{align}
        &\|\tilde{y}(k) - y(k)\| =  \left\| \xi\left(x_0(k+1), t\right) - \xi\left(x_0(k) + \Phi^{-1}(x_t(k), sp w(k), t), t\right)\right\| \nonumber \\
        &\leq \eta_2(t) \|x_0(k+1) - x_0(k) - \Phi^{-1}(x_t(k), sp w(k), t) \| \nonumber \\
        &= \eta_2(t) \left\| \sum_{l=1}^p \NNi{x_t(k)+(l-1)sw(k)}{sw(k)} \right.\nonumber\\
        &\hspace{10ex} \left. - \sum_{l=1}^p \Phi^{-1}(x_t(k)+(l-1)sw(k), s w(k), t)\right\|\nonumber\\
        &\leq p \eta_2(t)\max_{l=1,\ldots,p} \left\| \NNi{x_t(k)+(l-1)sw(k)}{sw(k)} \right.\nonumber\\
        &\hspace{20ex} \left.- \Phi^{-1}(x_t(k)+(l-1)sw(k), s w(k), t) \right\|  \label{eq:bounddiff}
    \end{align}
    }%
    where the equality is the third line is obtained by using \eqref{eq:x0increment} and rewriting $\Phi^{-1}(x_t(k),spw(k),t)$ as a telescoping sum. To bound the terms under the maximum in  \eqref{eq:bounddiff}, we now use that $N_{\Phi^{-1}}$ is an $(\epsr,\epsa)$-approximator of $\Phi^{-1}$. 
    
    By an application of Definition \ref{def:nn-approx-defn} followed by the lower bound in \eqref{eq:discr}, we obtain
    \begin{equation}
    \scriptsize
    \begin{aligned}
        \|\NNi{x}{v} - \Phi^{-1}(x, v, t)\| &\leq \epsr \|\Phi^{-1}(x, v, t)\| + \epsa  \\
        &\leq \epsr\eta_1(t)^{-1} \|v\| + \epsa
        \end{aligned}
    \end{equation}
    as long as $x \in \bD$, $\|v\| \leq r$ and $t \in [0,T]$. 
  
    Our assumptions imply that $s d_{init} \leq r$. Hence, whenever $\|w(k)\| \leq d_{init}$,  
    we have 
    \begin{equation}
    \scriptsize
        \label{eq:wdiffinal}
        \| \tilde{y}(k) - y(k) \| \leq s p\|w(k)\| \epsr \eta_2(t)  \eta_1(t)^{-1} + p \eta_2(t) \epsa 
    \end{equation}
    We will now use the above estimate to obtain a contraction-like argument.

    From Equation \ref{eq:walternate} we have
    {\scriptsize
    \begin{align*}
        w(k+1) - w(k) &= -x_t(k+1) + x_t(k) \doteq - \tilde{y}(k) = -y(k) + y(k) - \tilde{y}(k) \\
                      &= -spw(k) + y(k) - \tilde{y}(k). 
    \end{align*}}
    Note $\|w(k)\| = \| z - x_t(k) \| = d_a(k)$. Combining the above display with Equation \ref{eq:wdiffinal} establishes the following recursive inequality for $d_{a}(k)$ whenever $d_{a}(k) \leq d_{init}$:
    {\scriptsize
    \begin{equation*}
        \begin{aligned}
            d_a(k+1) &= \| w(k+1) \| = \| (1-sp)w(k) + y(k) - \tilde{y}(k) \| \\
            &\leq (1-sp)\|w(k)\| + \| \tilde{y}(k) - y(k)\| \\
           &\leq (1-sp\{1-\epsr \eta_2(t) \eta_1(t)^{-1}\})d_a(k) + p \eta_2(t) \epsa\\
            &= (1-sp\gamma_{\beps}(t)) d_a(k) + p\eta_2(t) \epsa
        \end{aligned}
    \end{equation*}}
    \noindent{%
    where we have used the assumption $sp \leq 1$ in second line,  \eqref{eq:wdiffinal} in the third line, and $\gamma_{\beps}(t) = 1 - \epsr \eta_2(t)\eta_1(t)^{-1}$ in the fourth line. From the assumed lower bound on $s$, we have $p \eta_2(t)\epsa \leq s p \gamma_{\beps}(t) d_{init}$, and the hence the condition $d_a(k) \leq d_{init}$ continues to holds for each $k \in \mathbb{N}$ by induction. By repeatedly applying the above inequality, we obtain
    \begin{equation}
    \scriptsize
       \begin{aligned}
       d_a(k) &\leq (1-sp\gamma_{\beps}(t))^k d_a(0)  + p \eta_2(t) \epsa\sum_{i=0}^{k-1} (1-sp\gamma_{\beps}(t))^{k-1} \\
       &\leq (1-sp\gamma_{\beps}(t))^k d_{init} + \frac{r_{\beps}(t)}{s}.
       \end{aligned}
    \end{equation}
    Since $sp \gamma_{\beps}(t) \in [0,1)$, we used the formula for the infinite geometric sum to obtain the last inequality.
    }

\end{proof}


\subsection{Guidance on designing better approximators}
\label{subsec:dir-based-estimator}
Theorem \ref{thm:bound-2} also provides guidance on how to design good approximators to use with \reachDest. For various approximators which one may consider that satisfy Definition \ref{def:nn-approx-defn}, $\epsa$ will typically be non-zero. Therefore, Theorem \ref{thm:bound-2} suggests that approximators with small additive error $\epsa$ will have better reachability guarantees when used within \reachDest{}. This naturally raises the question of how to design approximators with a small additive error $\epsa$.  One important aspect of this is the evaluation radius $r > 0$. For any given approximator $N_{\Phi^{-1}}$, the additive error $\epsa=\epsa(r)$ in Definition \ref{def:nn-approx-defn} can be considered as an increasing function of the evaluation radius $r$. Therefore, one may hope to obtain estimators with better values of $\epsa(r)$ by evaluating for small values of the radius $r$. 

In Figure \ref{fig:deltas}, we used test trajectories (i.e. system trajectories generated independently of the training data) to empirically estimate the absolute error $\epsa(r)$ for the approximator used in \neuralex evaluated for various values of $r$ on four systems. Even as $r \to 0$, the $\epsa(r)$ values of \neuralex seem to approach a non-zero value $\delta_0 = \lim_{r \to 0} \epsa(r) > 0$. In the light of Theorem \ref{thm:bound-2}, this might explain the lack of convergence of \neuralex that we have observed in certain empirical examples. Indeed, as the iterates $x_t^k$ in the \neuralex approach the target $z$, the error in the approximation of $N_{\Phi^{-1}}$ might possibly be dominating  the increment $\Phi^{-1}(x_t, s(z-x_t), t)$ needed to proceed towards the target.

Motivated by this discussion, in this work, we introduce approximators based on neural networks $\tilde{N}_{\Phi^{-1}}(x_t, v/\|v\|, t)$ that learn  only the direction (and not the magnitude) of the vector $\Phi^{-1}(x_t, v, t) \approx \nabla_v \Phi^{-1}(x_t,0,t) v$, for small values of $\|v\|$. By focusing only on learning the direction, we avoid numerical issues involved in learning  small vectors. Intuitively, if the value $\|\Phi^{-1}(x_t,v,t)\|$ was then known, we could use the oracle-estimator
\begin{equation}
\label{eq:oracle}
  N_{\Phi^{-1}}(x_t,v,t) = \tilde{N}_{\Phi^{-1}}(x_t, \frac{v}{\|v\|}, t)  \|\Phi^{-1}(x_t,v,t)\|
\end{equation}
to approximate $\Phi^{-1}(x_t,v,t)$ for small values of $\|v\|$. Figure \ref{fig:deltas} shows the estimate of the $\epsa(r)$ versus $r$ plot for the oracle estimator in \eqref{eq:oracle}. However, note that outside a testing scenario like that in Figure \ref{fig:deltas}, $N_{\Phi^{-1}}(x_t,v,t)$ cannot be evaluated for the directional approximators. Instead, we directly use $\tilde{N}_{\Phi^{-1}}(x_t,\frac{v}{\|v\|},t)$ in \reachDest algorithm by the modifications mentioned in Remark \ref{subsec:dir-based-est}.

\begin{figure}[ht]
\centering     
\includegraphics[width=90mm]{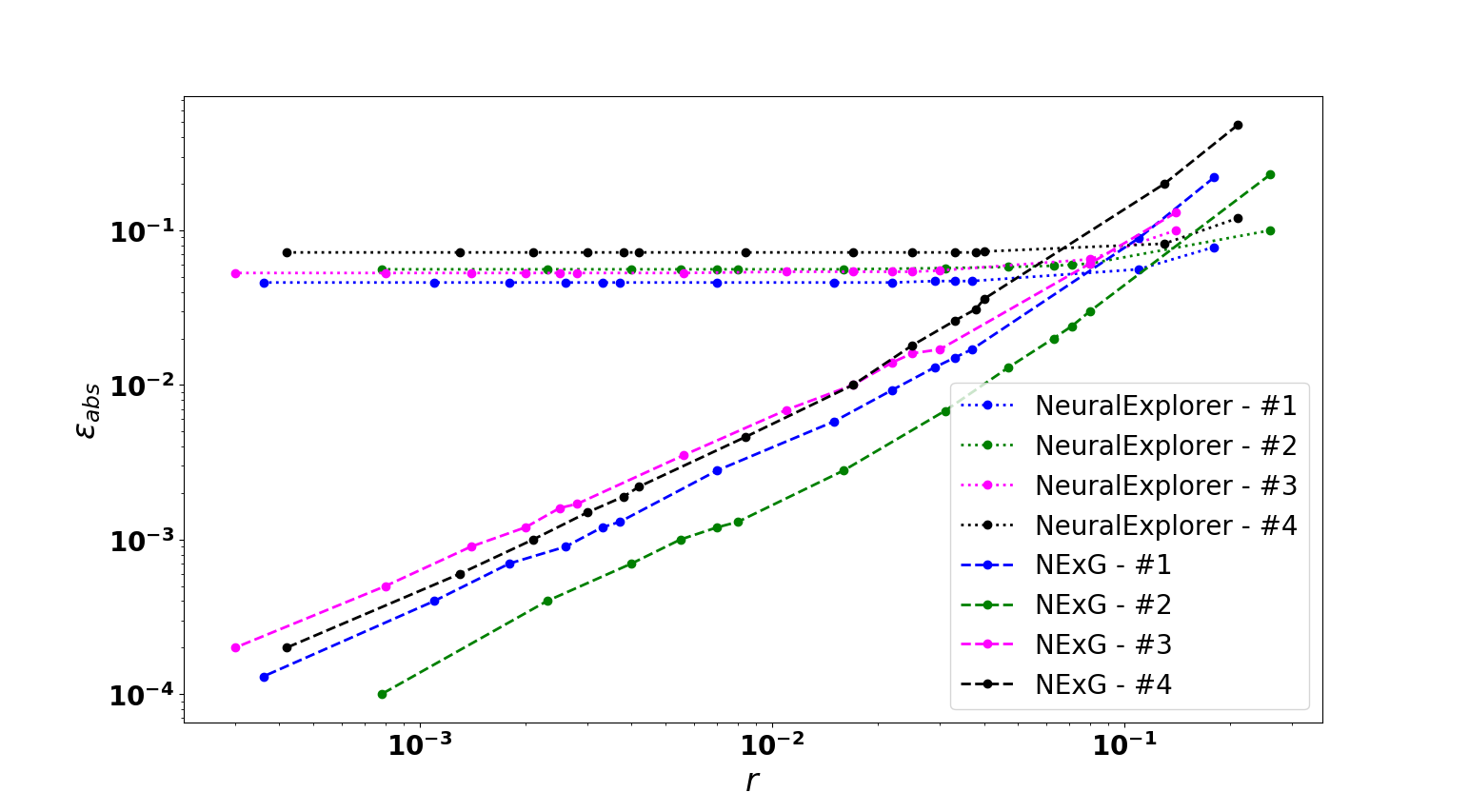}
\caption{\small{Empirical values of the additive error $\epsa$ in Definition \ref{def:nn-approx-defn} (assuming $\epsr \approx 0$) for the approximators $N_{\Phi^{-1}}$ learned for \neuralex and \nexg as a function of the evaluation radius $r$, for four benchmark systems. Note that we have the used the oracle-estimator $N_{\Phi^{-1}}$ given by \eqref{eq:oracle} to estimate the additive error of \nexg, since \nexg only learns a directional approximator $\tilde{N}_{\Phi^{-1}}$.}}

\label{fig:deltas}
\end{figure}

\begin{remark}
\label{subsec:dir-based-est}
%
As mentioned above, it may be helpful to work with directional-approximators $\tilde{N}(x_t,v/\|v\|, t)$ for $\Phi^{-1}(x_t, v, t)$, which learn only the direction (and not the magnitude) of the vector $\Phi^{-1}(x_t, v, t) \approx \nabla_v \Phi^{-1}(x_t, v, t)v$ for small values of $\|v\|$. One may then  work directly with the the directional approximator $\tilde{N}(x_t, v/\|v\|, t)$ in \reachDest by simply replacing 
 line~\ref{ln:predict} with the two statements
\begin{enumerate}
\item[$(8a)$] $\hat{v}_{-}^k \gets \tilde{N}_{\Phi^{-1}}(x_t^k, \frac{v^k}{\|v^k\|}, t)$ 
\item[$(8b)$] $\hat{v}_{-}^k \gets (s \times \|v^k\|) \cdot \hat{v}_{-}^k$,
\end{enumerate}
while keeping the rest of the algorithm unchanged.
\end{remark}

The \nexg algorithms used in the subsequent sections are based on directional-approximators $\tilde{N}_{\Phi^{-1}}$ and the corresponding modifications to \reachDest as mentioned in Remark \ref{subsec:dir-based-est}.

\section{Evaluation}
\label{sec:evaluation}

We choose a standard benchmark suite of control systems with neural feedback functions~\cite{10.1145/3302504.3311807,Tora_1996,ARCH20:ARCH_COMP20_Category_Report_Artificial,ARCH19:Verification_of_Closed_loop_Systems} for evaluation. To be more specific, systems \#10-\#12 are adopted from~\cite{Tora_1996}, system \#13 is adopted from~\cite{10.1145/3302504.3311807} and the rest of the benchmarks are adopted from the ARCH  suite~\cite{ARCH20:ARCH_COMP20_Category_Report_Artificial, ARCH19:Verification_of_Closed_loop_Systems}. Considered benchmarks span 6 dimensional systems, controllers with 6-10 hidden layers and 100-300 neurons per layer (c.f. Table~\ref{tab:learning-inverse-sensitivity}). The tool and the artifacts are available at  \url{https://github.com/manishgcs/NExG}.

\subsection{Network architecture and Training}
\label{sec:trainingNNs}

%
%
For each benchmark, we sample a fixed number (40) of initial states chosen uniformly at random and use a given ODE solver to generate \emph{anchor} trajectories from these initial states.
We further generate ten additional trajectories from the states randomly sampled in the small neighborhood ($\lVert v \rVert = 0.01$) of each initial state. We choose a step size ($h = 0.01$). Max step $(T)$ for each system is shown in Table ~\ref{tab:learning-inverse-sensitivity}, however, a user can pick any suitable values for these parameters including the number of anchor trajectories.
%
Our preliminary analysis shows that increasing the number of \emph{anchor} trajectories from $40$ to $50$ slightly improves the MRE, however, this improvement comes at the expense of higher training time. 
These trade offs between the amount of data required, training time, distance between neighboring points, and accuracy of the approximation are subjected to future research.
Nonetheless, the evaluations presented in subsequent sections  underscore the promise of our approach even when the resources are constrained.
The data used for training the neural network is collected as previously described. We use 90\% of the data for training and 10\% for testing.

\begin{table}[!ht]
\scriptsize
\centering
\caption{\textbf{Training $N_{\Phi^{-1}}$.}
Each neural network feedback controller configuration is given as the \emph{number of hidden layers} and the maximum of \emph{neurons per layer}. \emph{Dims} is the number of system variables and $T$ is simulation time bound. The training performance is measured in \emph{mean squared error} (MSE) and \emph{mean relative error} (MRE).}
\label{tab:learning-inverse-sensitivity}
\begin{tabular}{*{8}{c}}
\hline
\multicolumn{4}{|c|}{\emph{System}} & \multicolumn{4}{|c|}{$N_{\Phi^{-1}}$ Training} \\
& & & NN & Max & Training & & \\
No. & Name & Dims & controller & steps & Time & MSE & MRE \\
 & & & config & (T) & (min) & & \%\\

\hline
\hline
\#1 & \textsc{Arch}-1 & 2 & 6/50 & 200 & 8.0 & 0.018 & 16.0\\
\hdashline
\#2 & \textsc{Arch}-2 & 2 & 7/100 & 300 & 11.0 & 0.038 & 15.0\\
\hdashline
\#3 & \textsc{Arch}-3 & 2 & 5/50 & 350 & 14.0 & 0.014 & 10.2\\
\hdashline
\#4 & S. Pend. & 2 & 2/25 & 250 & 9.0 & 0.021 & 15.0\\
\hdashline
\#5 & \textsc{Arch}-4 & 3 & 7/100 & 250 & 10.0 & 0.007 & 6.0\\
\hdashline
\#6 & \textsc{Arch}-5 & 3 & 7/100 & 250 & 10.0 & 0.009 & 13.0\\
\hdashline
\#7 & \textsc{Arch}-6 & 3 & 6/100 & 250 & 10.0 & 0.007 & 5.6\\
\hdashline
\#8 & \textsc{Arch}-7 & 3 & 2/300 & 250 & 11.0 & 0.01 & 9.5\\
\hdashline
\#9 & \textsc{Arch}-8 & 4 & 5/100 & 250 & 12.0 & 0.005 & 8.0\\
\hdashline
\#10 & \textsc{Arch}-9-I & 4 & 3/100 & 250 & 12.0 & 0.005 & 4.2\\
\hdashline
\#11 & \textsc{Arch}-9-II & 4 & 3/20 & 250 & 11.0 & 0.005 & 7.3\\
\hdashline
\#12 & \textsc{Arch}-9-III & 4 & 3/20 & 250 & 11.0 & 0.005 & 4.5\\
\hdashline
\#13 & Unicycle & 4 & 1/500 & 250 & 11.0 & 0.005 & 5.3\\
\hdashline
\#14 & D. Pend.-I & 4 & 2/25 & 250 & 11.0 & 0.006 & 8.1\\
\hdashline
\#15 & D. Pend.-II & 4 & 2/25 & 200 & 10.0 & 0.02 & 17.0\\
\hdashline
\#16 & I. Pend. & 4 & 1/10 & 200 & 9.0 & 0.007 & 9.0\\
\hdashline
\#17 & ACC-3L & 6 & 3/20 & 200 & 11.0 & 0.004 & 8.3\\
\hdashline
\#18 & ACC-5L & 6 & 5/20 & 200 & 11.0 & 0.004 & 6.7\\
\hdashline
\#19 & ACC-7L & 6 & 7/20 & 200 & 12.0 & 0.004 & 6.7\\
\hdashline
\#20 & ACC-10L & 6 & 10/20 & 200 & 12.0 & 0.005 & 12.0\\
\hline
\end{tabular}

\end{table}

We use Python Multilayer Perceptron implemented in \texttt{Keras 2.3}~\cite{chollet2015keras} library with Tensorflow as the backend. Every network has 3 layers with 512 neurons each and an output layer. The input  layer's activation function is \emph{Radial Basis Function} (RBF) with \emph{Gaussian basis}~\cite{rbfkeras2019}.
The other two layers have \texttt{ReLU} activation (except System \#11 which has its feedback controller trained  with \texttt{Sigmoid} activation) and the output layer has \emph{linear} activation function.
The optimizer used is stochastic gradient descent. 
The network is trained using \textrm{Levenberg-Marquardt} backpropagation algorithm optimizing the mean absolute error loss function and the Nguyen-Widrow initialization.
The training and evaluation are performed on a system running Ubuntu 18.04 with a 2.20GHz Intel Core i7-8750H CPU with 12 cores and 32 GB RAM. 
The network $N_{\Phi^{-1}}$ training time, \emph{mean squared error} (MSE) and \emph{mean relative error} (MRE) for learning $\Phi^{-1}$ are given in Table~\ref{tab:learning-inverse-sensitivity}. 

Although empirical, our choice for network architecture and evaluation metrics are somewhat motivated  by \neuralex~\cite{DBLP:conf/atva/0002D20}. 
But the network training time in~\cite{DBLP:conf/atva/0002D20} and training error are notably high,  which may render the work unfavorable for many practical applications. 
After performing experiments on multiple activation functions, we choose a non linear multi-variate radial basis function (RBF) with Gaussian basis as the input layer's activation function because we observe that its evaluation performance is consistent across benchmarks.

%

\subsection{\textsf{ReachDestination} Evaluation}

We analyze the performance of Algorithm~\ref{alg:reachDest} by picking, at every invocation of the algorithm, a random reference trajectory $\xi_{A}$, a time $t \in [0,T]$, and a target state $z$, reachable at time $t$ in the domain of interest.
We choose them randomly to not bias the evaluation of our search procedure to a specific sub-space.
The performance metrics used to evaluate various runs are \emph{number of course corrections} ($k$) and/or \emph{minimum relative distance} ($d_r$). The threshold $\delta$ is fixed as 0.004. 
%
%


    \textbf{B.1 Comparison with \texttt{NeuralExplorer}~\cite{DBLP:conf/atva/0002D20}:} 
    The neural network architectures used in this work are the same as those used in \neuralex.
    For a given $N \in \mathbb{Z}_+$ number of anchor trajectories, \neuralex creates all possible $C(N, 2)$ pairs of these trajectories for training as it attempts to learn the inverse sensitivity function for any $v \in \reals^n$ in the domain of interest. Whereas \nexg focuses  on learning only the direction of the inverse sensitivity.  So we only sample a few random points (say, $y$) in a small neighborhood of the initial state of each anchor trajectory and generate total $y \times N$ pairs. As a consequence, we achieve up to 60\% reduction in the training time. 
    Further, the state space exploration algorithm in \neuralex predicts inverse sensitivity directly for $w^k$ and course corrects at every step; whereas, \nexg predicts only the direction of the inverse sensitivity vector needed to move in the direction $w^k$. Hence the \nexg  search is guided by additional parameters like the scaling factor $s$ and the correction period $p$.
    We report in Table~\ref{tab:result-reachDest-wrt-nerualex} the mean values of $k$ and $d_r$ computed over $250$ runs of each technique for each system. 
    The evaluation shows that \nexg has a relative error of 1-4\% (with considerably fewer iterations) as compared to the relative error of 5-15\% for \neuralex.
    %
    %
    \begin{table}
\scriptsize
\centering
\caption{\textbf{Performance evaluation w.r.t. NeuralExplorer.} The common parameters values are $\delta=0.004$ and $\mathcal{B}=30$. We fix $s=0.5$ and $p=2$ for NExG. $k$ is the number of simulations generated by respective algorithms and $d_r\%=(d_a^k/d_{init})\times 100$ where $d_a^k$ is the distance between the state of final ($k^{th}$) simulation at time $t$ and the destination.}
\label{tab:result-reachDest-wrt-nerualex}
\begin{tabular}{|*{10}{ccccc||ccccc|}}
\hline
Sys. & \multicolumn{2}{c|}{NeuralExp.} & \multicolumn{2}{c||}{NExG} & Sys. & \multicolumn{2}{c|}{NeuralExp.} & \multicolumn{2}{c|}{NExG} \\
& $k$ & $ d_r\%$ & $k$ & $ d_r\%$ & & $k$ & $ d_r\%$ & $k$ & $ d_r\%$\\
\hline
\hdashline
\hline
\#1 & 21 & 14 & 5 & 1.8 & \#11 & 30 & 6.1 & 4 & 0.4  \\
\hdashline
\#2 & 27 & 9.6 & 6 & 1.3 & \#12 & 30 & 13.8 & 4 & 0.6 \\
\hdashline
\#3 & 10 & 6.4 & 5 & 2.7 & \#13 & 30 & 11.4 & 13 & 2.3 \\
\hdashline
\#4 & 18 & 5.5 & 5 & 1.4 & \#14 & 29 & 7.6 & 10 & 1.9 \\
\hdashline
\#5 & 30 & 13.3 & 7 & 2.7 & \#15 & 26 & 12.7 & 8 & 3.7 \\
 \hdashline
 \#6 & 24 & 11.2 & 5 & 3.9 & \#16 & 29 & 22.5 & 6 & 2.2 \\
 \hdashline
 \#7 & 28 & 12.5 & 5 & 2.1 & \#17 & 30 & 8.0 & 8 & 0.4 \\
 \hdashline
\#8 & 23 & 5.5 & 7 & 1.7 & \#18 & 30 & 6.7 & 7 & 0.4 \\
 \hdashline
\#9 & 30 & 12.3 & 12 & 3.5 & \#19 & 30 & 5.5 & 6 & 0.3 \\
 \hdashline
\#10 & 29 & 4.3 & 10 & 1.0 & \#20 & 30 & 5.4 & 17 & 1.6 \\
  \hline
\end{tabular}

\end{table}
    
    \textbf{B.2 Correction period ($p$) and scaling factor ($s$):} 
    We fix the tuple ($\xi_{A}$, $z$, $t$) in each benchmark, run \reachDest for $s \in \{0.01, 0.1\}$, $p \in \{1, 5, 10\}$. The evaluation results in Table~\ref{tab:result-reachDest-partial} are presented to make some key observations and emphasize that the technique performs consistently across systems.
    For a fixed tuple ($\xi_{A}$, $z$, $t$), change in the number of trajectories ($k$) generated by \reachDest is roughly inversely proportional to the change in the product $s \cdot p$. For example, first row (i.e., $s=0.01$) in System \#7 
    shows that the number of trajectories ($k$) reduces from $\sim 400$ to $\sim 40$ (10 fold reduction) when course correction is performed only once for every $10$ steps instead of at every steps. This trend is observed in almost all systems for appropriate $s \cdot p$ values.
    It can also be observed that the number of course corrections ($k$) remains roughly the same for different ($s$, $p$) pairs as long as the product $s \cdot p$ is same. For e.g., the value of $k$ for pairs ($s=0.01$, $p=10$) and ($s=0.1$, $p=1$) is $\sim 40$ in System \#1. These results are consistent with the theoretical bound (Equation~\ref{eq:thm-bound}) that decreases geometrically in $k$ for a fixed value of $s \cdot p$.

\begin{table}[tp]
\footnotesize
\centering
\setlength{\tabcolsep}{10pt}
\caption{\reachDest \textbf{evaluation}. $d_{init}$ is the distance between the state of the initial reference trajectory at time $t$ and destination $z$, $s$ is the scaling factor, and $p$ is the correction period.}
\label{tab:result-reachDest-partial}
\begin{tabular}{|*{6}{c||c|c|c|c|c|} }
\hline
System & $d_{init}$ & s & \multicolumn{3}{c|}{Course corrections $k$} \\
& & & $p = 1$ & $p = 5$ & $p = 10$ \\
\hline
\hline
  \multirow{2}{*}{\#1} &  \multirow{2}{*}{0.43} & 0.01 & 528 & 105 & 52 \\
  \cdashline{3-6}
  & & 0.1 & 52 & 10 & 3 \\
 \hline
 \multirow{2}{*}{\#7} &  \multirow{2}{*}{0.39} & 0.01 & 418 & 83 & 41 \\
 \cdashline{3-6}
 & & 0.1 & 41 & 7 & 3 \\
  \hline
  \multirow{2}{*}{\#11} &  \multirow{2}{*}{0.37} & 0.01 & 381 & 76 & 37  \\
  \cdashline{3-6}
 & & 0.1 & 37 & 7 & 3 \\
  \hline
\multirow{2}{*}{\#17} &  \multirow{2}{*}{0.85} & 0.01 & 550 & 109 & 54 \\
 \cdashline{3-6}
 & & 0.1 & 54 & 10 & 3\\
   \hline
\end{tabular}

\end{table}

     \textbf{B.3 Satisfying initial conditions:} As the algorithm increments the initial state $x_0^k$ in line~\ref{ln:perturbx}, it may happen that next $\hat{x}_0^{k+1} \doteq (x_0^k + \hat{v}_{-}^k)$ is not in the initial set $\theta$, thus violating the initial constraint.
    %
    We address this problem by picking a element-wise projection of $\hat{x}^k_0$ in $\theta$ denoted as $\hat{x}_{0}^{k, \theta} \doteq proj_{\theta}(\hat{x}_0^k) \in \theta$, defined by  $\hat{x}_{0}^{k, \theta} \doteq \argminA_{x \in \theta} \lVert x - \hat{x}_0^k \rVert$.
    Consider System \#10 with $2^{nd}$ component of its initial set hyper-rectangle, given as    $\mathrm{\theta[2]} \doteq [-0.5, 0.0]$. Both Figures ~\ref{fig:bench9-w-adapt-1-01} and ~\ref{fig:bench9-w-adapt-greedy-01} demonstrate how the course of exploration makes a detour around the initial set boundary in order to satisfy its constraints.

    \begin{figure}[!ht]
\centering     
\subfigure[Original \reachDest ]{\label{fig:bench9-w-adapt-1-01}\includegraphics[width=70mm, height=45mm]{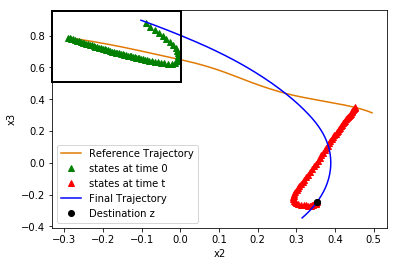}}
\subfigure[Axis-aligned \reachDest ]{\label{fig:bench9-w-adapt-greedy-01}\includegraphics[width=70mm, height=45mm]{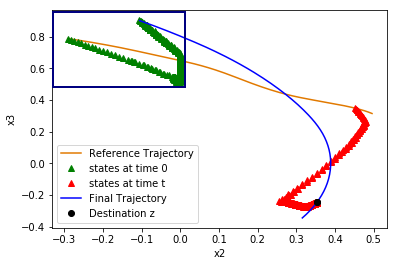}}

\caption{\small{\reachDest and its customizations can  provide algorithms for state space exploration with  a constrained initial set. In the figure, the inner  box represents the initial set.
Original \reachDest generates smoother trajectories because it moves in the direction of the target at each step (Figure~\ref{fig:bench9-w-adapt-1-01}), however \reachDest can be customized to obtain a different state space exploration method.}}
\label{fig:features}
\end{figure}
    
    \textbf{B.4 Customizing the state exploration algorithm:}
    Note that the inverse sensitivity approximator $N_{\Phi^{-1}}$ is agnostic to the exact state space exploration technique. While our implementation of \reachDest uses this estimator to proceed in a straight line direction towards the destination (i.e. $v^k$ has the same direction as $z - x_t^k$), the progress direction can also be customized. This allows for  designing custom state space exploration algorithm by prioritizing trajectories along different directions at different steps.
    For an $n$-dimensional system, at every step, one might be interested in picking a direction among the $2n$ unit vectors $\{\pm e_i : i = 1, 2, \ldots, n\}$ that are aligned with the orthonormal axes. For instance, one can choose the direction vector that is closest to $z-x_t$. 
    The illustration of one such \emph{axis-aligned} approach is given in Figure~\ref{fig:bench9-w-adapt-greedy-01}. It emphasizes that instead of \reachDest, we can also use some other state space exploration algorithm that requires an inverse sensitivity approximator.
    
    \textbf{B.5 Coverage analysis:}  Given an initial set $\theta \subseteq \bD$, we assess the coverage among the set of reachable states at time $t \in [0, T]$ by calculating the proportion of points in the reachable set that \reachDest converges to, within a neighborhood of radius $\delta$.
    To obtain a convenient representation of the reachable set for an $n$-dimensional system, we use a polygon with faces in the $2n$ template directions $\{\pm e_i: i = 1, 2, \ldots, n\}$. While we have used orthonormal vectors as template directions, different set of template directions can yield a less conservative approximation of the reachable set.  
     The polygon in our experiment was obtained by starting from the destination state of random anchor trajectory $\xi_A(\cdot)$ at time $t$, and using a modification of \reachDest to maximally perturbs the destination state in each of the template directions. This provides as many extremal points as the number of template directions, and can be used to construct the bounding polygon (e.g. see the black rectangle in Figure \ref{fig:coverage-b1}), denoted by $\mathcal{Z}$,  as an approximation of the reachable set.  
    Next, to assess the coverage for $\mathcal{Z}$, we sampled $200$  points from  $\mathcal{Z}$  uniformly at random, and examined which were the ones that \reachDest could converge within a  $\delta=4\times 10^{-3}$ neighborhood at time $t$ starting from the initial set $\theta$. As shown in  
     Figure~\ref{fig:coverage1-b1}, 137 out of these 200 points were reached from \reachDest, with the color of the point (green or red) representing if \reachDest was successful or not. For each of these points in $\mathcal{Z}$, we also plot the best initial point output by \reachDest. Most of these red initial points (that did not reach the destination) lie on the boundary of the initial set $\theta$, suggesting that the trajectory that can possibly reach its destination might perhaps start from a state outside the given initial set. However, a less conservative approximation of the reachable set obtained by picking a different set of template directions yields much better coverage as demonstrated in Figure~\ref{fig:coverage2-b1}.

    \begin{figure}[H]
\centering     
\subfigure[Coarse approximation]{\label{fig:coverage1-b1}\includegraphics[width=75mm, height=45mm]{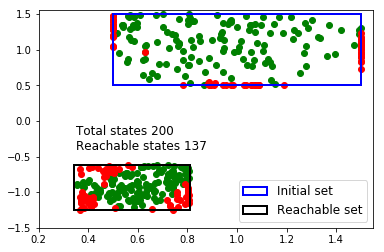}}
\subfigure[Less conservative  approximation]{\label{fig:coverage2-b1}\includegraphics[width=75mm, height=45mm]{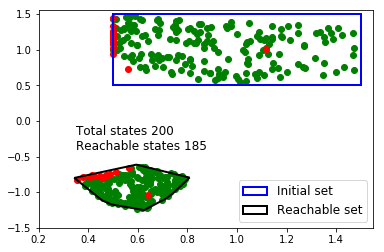}}
\caption{\small{Measuring coverage of a reachable set in System \#1. For every red colored state in the destination set, \reachDest could not find a trajectory that reaches within its $\delta$-neighborhood at time $t$.}}
\label{fig:coverage-b1}
\end{figure}

\section{Falsification of a Safety Specification}
\label{sec:falsification}

Given system and a corresponding safety specification in either Signal or Metric Temporal Logic~\cite{STLmain,MTLmain}, \emph{falsification} is aimed at finding a system parameter or an input that violates the specification.
Existing falsification schemes generate executions using some heuristics or stochastic global optimization and compute their robustness with respect to a safety specification provided as a set of states. 
\emph{Robustness} ($\rho \in \mathbb{R}$) is a measure that quantifies how deep is the execution within the set or how far away it is from the set. Informally, it determines the degree to which an execution satisfies ($\rho > 0$) or violates ($\rho < 0$) a given safety specification. Our framework can currently handle a subset of MTL formulas.

\begin{equation}
\label{eqn:mtlformula}
\varphi ::= \top ~|~ p ~|~ \neg \varphi ~|~ \top \mathcal{U}_l \varphi
\nonumber
\end{equation}

\noindent{where $p$ is an atomic proposition, $l$ is a non-empty interval of $\mathbb{R}_{+}$, and $\varphi$ is a well formed MTL formula. The temporal operator  $\diamond$ $(eventually)$ is defined as $\diamond_{l}\varphi := \top \mathcal{U}_l \varphi$}. The reader can refer to~\cite{MCFalsification} for robust semantics of MTL formulas.

\subsection{Our Falsification algorithm}
We describe a simple \reachDest-based algorithm to obtain a falsifying trajectory to a given safety specification $\neg\Diamond_{l}U$, where $U \subseteq \mathbb{R}^n$ is the unsafe set.
We generate an anchor trajectory $\xi_{A}$, sample a state $z \in U$, and choose $t = \argminA_{t'\in l}\lVert \xi_{A}(t')-z\rVert$.
We then invoke \reachDest sub-routine for generating trajectories until we obtain a counterexample ($\rho^k < 0$) to the given safety specification or bound $I$ is exhausted, where $\rho^k$ is the robustness of trajectory $\xi_A^{k}$. To be precise, in the  falsification run of \reachDest - (i) distance $d_a^k$ is replaced by robustness $\rho^k$, (ii) constraint $d_a^k > \delta$ is replaced by $\rho^k > 0$, and (iii) an additional constraint $x_t^k \notin U$ is added to the main \textbf{while} loop condition.
While it may be the case that $z$ is not reachable at time $t$, both these parameters primarily act as anchors to guide the procedure in obtaining a falsifying execution.

\subsection{Evaluation of Falsification techniques}
We evaluate our falsification algorithm against one of the widely used \emph{falsification} platforms, S-TaLiRo~\cite{MCFalsification}. 
Monte-Carlo sampling scheme in S-TaLiRo is sensitive to the ``temperature" parameter $\beta$, where the adaptation of $\beta$ is performed after every fixed number of iterations provided it is unable to find a counterexample by then.
We keep $\beta = 50$ which is the default value, and we consider $p=2, s=0.5$ for our \reachDest-based falsification scheme.  
%
Although adaptation parameters and mechanisms in both approaches are different, an upper bound ($\mathcal{B}$) on the number of trajectories is crucial to both of them. We fix $\mathcal{B}=100$ for systems \#1-\#16 and $\mathcal{B}=150$ for systems \#17-\#30 in S-TaLiRo. We consider $\mathcal{B}=50$ for \nexg as we notice that, if it can, it usually finds a trajectory of interest in notably less number of iterations. The sampling time is fixed as $0.01$.
We exclude cases where the initial reference trajectory $\xi_A$ is falsifying so as to minimize the bias induced by different distributions in different schemes. 
For a given pair of initial configuration $\theta$ and safety specification $\neg\Diamond_{l}U$ in each system, we report in Table~\ref{tab:result-staliro} the \emph{mean} of total trajectories ($k$) generated along with \emph{mean} robustness ($\rho$) computed over $250$ runs of respective techniques.

\begin{figure}[!b]
\centering     

\subfigure[from S-TaLiRo]{\label{fig:bench9-sig-staliro}\includegraphics[width=70mm, height=40mm]{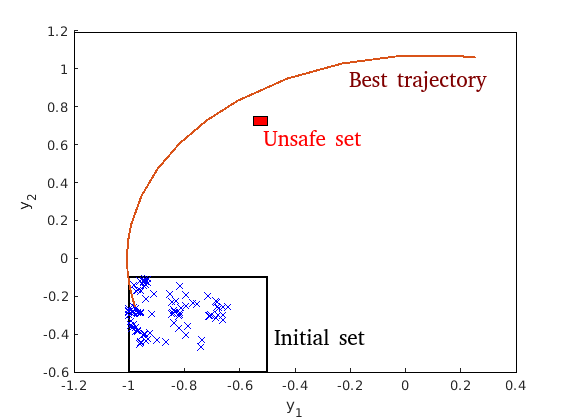}}
\subfigure[from NExG]{\label{fig:bench9-sig-ours}\includegraphics[width=68mm, height=40mm]{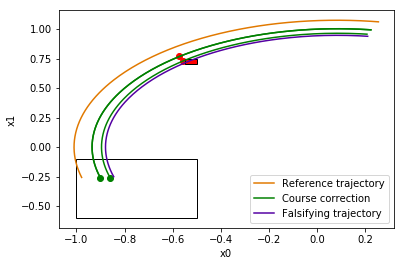}}
\label{fig:falsification-w-staliro}
\caption{\small{Falsification demonstrations. The red-colored box is the unsafe set and the inner while-colored box is the initial set. These demonstrations depict how NExG can potentially supplement other falsification platforms if they fail to find a falsifying execution.}}
\end{figure}

The evaluations exhibit that our algorithm not only takes a very few trajectories to converge but also  the counterexample obtained is relatively more robust in most cases. Unlike \textrm{NExG}, the performance of S-TaLiRo seems to deteriorate further with increase in the number of system dimensions and complexity. Even in scenarios where \reachDest generates more than $10$ trajectories, experiments indicate that it is able to reach the neighborhood around $U$ within fewer iterations. 
This observation motivated us to attempt to integrate both frameworks. In the case of non-convergence in S-TaLiRo, its best execution can be used as the input \emph{reference} trajectory  for $\mathcal{RD}$.
One such instance is shown in Figure~\ref{fig:bench9-sig-staliro} where S-TaLiRo is unable to find a falsifying trajectory within 100 iterations. 
We use its best sample as the reference $\xi_A$ for \reachDest and find $4^{th}$ trajectory to be a counterexample (Figure~\ref{fig:bench9-sig-ours}).
This exercise is performed for illustration purpose i.e., at present we manually port the best sample from S-TaLiRo to \nexg. One of the future tasks is to automate this integration. 
%
Additionally, our approach - as a side effect - provides intuition about the course of exploration leading to the falsifying execution unlike scattered stochastically sampled states generated in S-TaLiRo.

Another important take away from this comparison is that if S-TaLiRo fails to find a counterexample for a given specification, the user is left with a sample of trajectories generated by S-TaLiRo and the execution that comes closest to falsifying the given safety specification.
Instead, in our case, the user can still access the inverse sensitivity approximator and manually (or algorithmically) probe nearby trajectories and proceed to discover a falsifying trajectory. Finally, S-TaLiRo's implementation platform is \texttt{MATLAB} while our framework is implemented in \texttt{Python}. We do not report the wall-clock time taken by respective frameworks as performance differences are expected due to their different implementation platforms.

One may analogize inverse sensitivity approximator to an offline hash of results, which may make the comparison with S-TaLiRo (which performs Falsification in an online fashion) seem unfair. This analogy, however, motivates us to make some clarifications and highlight a few important differences as a rationale behind the comparison. 

Note first that, in the training phase, we learn the inverse sensitivity based on trajectories starting from uniformly sampled points in the domain. In particular, instead of learning a system property specific to the task at hand, we are learning a generic system property that can be used in many applications.

Secondly, our intention to perform this comparative analysis is not to compete with S-TaLiRo, but to emphasize  the effectiveness of learning the sensitivity function in an offline manner to aid in state space exploration. Standard gradient based techniques are not designed to use an offline component because they only rely on the current simulation to falsify the given specification based on robustness, and information from previous simulations is discarded at each step.

In our systematic state space exploration, simulation at each successive step is guaranteed (under some assumptions) to get closer to a desirable execution. We agree that we pre-process data to learn system properties, but this learned information enable us to deliver more promise by guaranteeing convergence as supported by our evaluations. 

While we have primarily focused on neural network feedback control systems for our work, the presented  approach can also be applied to other nonlinear dynamical and hybrid systems. Evaluations with a few standard higher dimensional non-linear dynamical systems yields similar performance results. The falsification results in S-TaLiRo for \textit{LaubLoomis} (7-dim), \emph{Biological model I} (7-dim), and \emph{Biological model II} (9-dim) are ($\rho = 0.003, k = 112$), ($\rho = 0.046, k = 149$), and ($\rho = 0.155, k = 150$); whereas, the respective results in \nexg are ($\rho = 0.003, k = 21$), ($\rho = -0.0016, k = 15$), and ($\rho = 0.002, k = 29$).

\begin{table*}[ht]
\setlength{\tabcolsep}{7pt}
\scriptsize
\centering
\caption{Performance of falsification techniques. $k$ is the number of simulations generated and $\rho$ is the robustness. The parity of $\rho$ determines whether the execution satisfies ($\rho > 0$) or falsifies ($\rho < 0$) a given safety specification, whereas its magnitude determines how robust is the execution. NExG takes a very few iterations to find a counterexample with $\rho < 0$. \bcheckmark marks the scenarios with equally (or more) robust falsifying trajectory.}
\label{tab:result-staliro}
\begin{tabular}{|*{7}{c|} }
\hline
 & Initial & Safety & \multicolumn{2}{c}{~~S-TaLiRo}& \multicolumn{2}{|c|}{NExG}\\
System & configuration & specification &  Mean & Mean &  Mean & Mean\\
& $\theta$ & $\neg\Diamond_l U$ & $k$ & $\rho$ & $k$ &  $\rho$\\
\hline
\#1 & [(0.5, 1.5)$\wedge$(0.5, 1.5)] & \textcolor{blue}{$\neg\Diamond_{[0.7, 0.9]}
$}[(0.30, 0.35)$\wedge$(-1.1, -1.05)] & 12 & -0.01\bcheckmark & 3 & -0.01\bcheckmark  \\
\hdashline
\#2 & [(0.8, 1.2)$\wedge$(0.9, 1.2)] & \textcolor{blue}{$\neg\Diamond_{[0.7, 0.9]}$}[(1.50, 1.55)$\wedge$(0.20, 0.25)] & 24 & -0.006\bcheckmark & 3 & -0.005  \\
\hdashline
\#3 & [(0.4, 1.2)$\wedge$(0.4, 1.2)] & \textcolor{blue}{$\neg\Diamond_{[0.8, 1.0]}$}[(0.3, 0.4)$\wedge$(0.0, 0.1)] & 8 & -0.02\bcheckmark & 4 & -0.014  \\
\hdashline
\#4 & [(1.5, 2.0)$\wedge$(1.0, 1.5)] & \textcolor{blue}{$\neg\Diamond_{[0.7, 0.9]}$}[(1.15, 1.2)$\wedge$(-0.95, -0.90)] & 11 & -0.008\bcheckmark & 3 & -0.008\bcheckmark  \\
\hdashline
\#5 & [(0.2, 0.7)$\wedge$(0.2, 0.7)$\wedge$(0.2, 0.7)] & \textcolor{blue}{$\neg\Diamond_{[0.6, 0.8]}$}[(1.0, 1.05)$\wedge$(0.05, 0.1)$\wedge$(-1.15, -1.10)] & 39 & 0.008 & 17 & -0.005\bcheckmark  \\
\hdashline
\#6 & [(0.1, 0.6)$\wedge$(0.1, 0.6)$\wedge$(0.1, 0.6)] & \textcolor{blue}{$\neg\Diamond_{[0.7, 0.9]}$}[(0.10, 0.15)$\wedge$(0.0, 0.05)$\wedge$(0.10, 0.15)] & 21 & -0.005 & 3 &  -0.007\bcheckmark  \\
\hdashline
\#7 & [(0.2, 0.5)$\wedge$(0.2, 0.5)$\wedge$(0.2, 0.5)] & \textcolor{blue}{$\neg\Diamond_{[1.0, 1.2]}$}[(0.4, 0.45)$\wedge$(-0.3, -0.25)$\wedge$(-0.45, -0.4)] & 20 & -0.005 & 3 &  -0.007\bcheckmark  \\
\hdashline
\#8 & [(0.2, 0.5)$\wedge$(0.2, 0.5)$\wedge$(0.2, 0.5)] & \textcolor{blue}{$\neg\Diamond_{[1.0, 1.2]}$}[(0.05, 0.1)$\wedge$(0.25, 0.3)$\wedge$(-0.35, -0.3)] & 36 & 0.008 & 9 &  0.001\bcheckmark  \\
\hdashline
\multirow{2}{*}{\#9} & [(0.1, 0.4)$\wedge$(0.1, 0.4) & \textcolor{blue}{$\neg\Diamond_{[0.6, 0.8]}$}[(-0.15, -0.10)$\wedge$(-0.80, -0.75) & \multirow{2}{*}{76} &  \multirow{2}{*}{0.005\bcheckmark} & \multirow{2}{*}{22} &  \multirow{2}{*}{0.005\bcheckmark}  \\
& $\wedge$(0.1, 0.4)$\wedge$(0.1, 0.4)] & $\wedge$(0.0, 0.05)$\wedge$(-0.60, -0.55)]& & &  &\\
\hdashline
\multirow{2}{*}{\#10} & [(0.5, 1.0)$\wedge$(-1, -0.5) & \textcolor{blue}{$\neg\Diamond_{[0.8, 1.0]}$}[(-0.2, -0.15)$\wedge$(-1.05, -1.0) & \multirow{2}{*}{88} &  \multirow{2}{*}{0.01} & \multirow{2}{*}{5} &  \multirow{2}{*}{-0.005\bcheckmark}  \\
& $\wedge$(-0.5, 0)$\wedge$(0.5, 1)] & $\wedge$(0.2, 0.25)$\wedge$(0.25, 0.3)]& & & &\\
\hdashline
\multirow{2}{*}{\#11} & [(-1, -0.5)$\wedge$(-0.6, -0.1) & \textcolor{blue}{$\neg\Diamond_{[1.0, 1.2]}$}[(-0.55, -0.5)$\wedge$(0.7, 0.75) & \multirow{2}{*}{95} &  \multirow{2}{*}{0.031} & \multirow{2}{*}{3} &  \multirow{2}{*}{-0.016\bcheckmark}  \\
& $\wedge$(0.2, 0.7)$\wedge$(-0.5, 0)] & $\wedge$(0.65, 0.7)$\wedge$(0.25, 0.3)]& & & &\\
\hdashline
\multirow{2}{*}{\#12} & [(-1, -0.5)$\wedge$(-0.6, -0.1) & \textcolor{blue}{$\neg\Diamond_{[1.0, 1.2]}$}[(-0.55, -0.5)$\wedge$(0.25, 0.3) & \multirow{2}{*}{98} &  \multirow{2}{*}{0.058} & \multirow{2}{*}{6} &  \multirow{2}{*}{-0.009\bcheckmark}  \\
& $\wedge$(0.2, 0.7)$\wedge$(-0.5, 0)] & $\wedge$(-0.05, 0.0)$\wedge$(-0.3, -0.25)]& & & &\\
\hdashline
\multirow{2}{*}{\#13} & [(9.3, 9.7)$\wedge$(-4.7, -4.3) & \textcolor{blue}{$\neg\Diamond_{[0.8, 1.0]}$}[(8.4, 8.45)$\wedge$(-3.55, -3.5) & \multirow{2}{*}{81} &  \multirow{2}{*}{0.009} & \multirow{2}{*}{4} &  \multirow{2}{*}{-0.01\bcheckmark}  \\
& $\wedge$(2, 2.4)$\wedge$(1.3, 1.7)] & $\wedge$(2.5, 2.55)$\wedge$(2.2, 2.25)]& & & &\\
\hdashline
\multirow{2}{*}{\#14} & [(1.0, 1.5)$\wedge$(1.0, 1.5) & \textcolor{blue}{$\neg\Diamond_{[1.0, 1.2]}$}[(1.55, 1.6)$\wedge$(0.25, 0.3) & \multirow{2}{*}{59} &  \multirow{2}{*}{-0.002\bcheckmark} & \multirow{2}{*}{6} &  \multirow{2}{*}{-0.002\bcheckmark}  \\
& $\wedge$(1, 1.5)$\wedge$(1, 1.5)] & $\wedge$(-0.65, -0.6)$\wedge$(-1.2, -1.15)]& & & &\\
\hdashline
\multirow{2}{*}{\#15} & [(1.0, 1.4)$\wedge$(1.0, 1.4) & \textcolor{blue}{$\neg\Diamond_{[0.4, 0.6]}$}[(0.90, 0.95)$\wedge$(0.65, 0.7) & \multirow{2}{*}{55} &  \multirow{2}{*}{-0.002} & \multirow{2}{*}{8} &  \multirow{2}{*}{-0.007\bcheckmark}  \\
& $\wedge$(1.0, 1.4)$\wedge$(1.0, 1.4)] & $\wedge$(-1.8, -1.75)$\wedge$(-1.20, -1.15)]& & & &\\
\hdashline
\multirow{2}{*}{\#16} & [(0.0, 0.3)$\wedge$(0.0, 0.3) & \textcolor{blue}{$\neg\Diamond_{[0.8, 1.0]}$}[(0.05, 0.1)$\wedge$(-0.05, 0.0) & \multirow{2}{*}{47} &  \multirow{2}{*}{-0.002} & \multirow{2}{*}{9} &  \multirow{2}{*}{-0.003\bcheckmark}  \\
& $\wedge$(0.0, 0.3)$\wedge$(-0.3, 0.0)] & $\wedge$(-0.05, 0.0)$\wedge$(0.0, 0.05)]& & & &\\
\hdashline
\multirow{2}{*}{\#17} & [(90.0, 92.0)$\wedge$(32.0, 32.5)$\wedge$(0.0, 0.0) & \textcolor{blue}{$\neg\Diamond_{[0.7, 0.9]}$}[(113.5, 114.0)$\wedge$(31.3, 31.4)$\wedge$(-1.60, -1.55) & \multirow{2}{*}{150} &  \multirow{2}{*}{0.10} & \multirow{2}{*}{15} &  \multirow{2}{*}{-0.001\bcheckmark}  \\
& $\wedge$(10.0, 11.0)$\wedge$(30.0, 30.5)$\wedge$(0.0, 0.0)] & $\wedge$(32.0, 32.5)$\wedge$(29.5, 30.0)$\wedge$(-0.10, -0.05)] & & & &\\
\hdashline
\multirow{2}{*}{\#18} & [(90.0, 92.0)$\wedge$(32.0, 32.5)$\wedge$(0.0, 0.0) & \textcolor{blue}{$\neg\Diamond_{[0.7, 0.9]}$}[(116.5, 117.0)$\wedge$(31.6, 31.7)$\wedge$(-1.65, -1.6) & \multirow{2}{*}{150} &  \multirow{2}{*}{0.81} & \multirow{2}{*}{14} &  \multirow{2}{*}{-0.002\bcheckmark}  \\
& $\wedge$(10.0, 11.0)$\wedge$(30.0, 30.5)$\wedge$(0.0, 0.0)] & $\wedge$(34.5, 35.0)$\wedge$(29.5, 30.0)$\wedge$(-0.45, -0.35)] & & & &\\
\hdashline
\multirow{2}{*}{\#19} & [(90.0, 92.0)$\wedge$(32.0, 32.5)$\wedge$(0.0, 0.0) & \textcolor{blue}{$\neg\Diamond_{[0.8, 1.0]}$}[(120.0, 120.3)$\wedge$(31.1, 31.2)$\wedge$(-1.75, -1.7) & \multirow{2}{*}{149} &  \multirow{2}{*}{0.11} & \multirow{2}{*}{9} &  \multirow{2}{*}{-0.007\bcheckmark}  \\
& $\wedge$(10.0, 11.0)$\wedge$(30.0, 30.5)$\wedge$(0.0, 0.0)] & $\wedge$(37.0, 38.0)$\wedge$(30.5, 31.0)$\wedge$(0.1, 0.2)] & & & &\\
\hdashline
\multirow{2}{*}{\#20} & [(90.0, 92.0)$\wedge$(32.0, 32.5)$\wedge$(0.0, 0.0) & \textcolor{blue}{$\neg\Diamond_{[0.6, 0.8]}$}[(112.8, 112.9)$\wedge$(31.7, 31.8)$\wedge$(-1.55, -1.5) & \multirow{2}{*}{145} &  \multirow{2}{*}{0.14} & \multirow{2}{*}{34} &  \multirow{2}{*}{0.015\bcheckmark}  \\
& $\wedge$(10.0, 11.0)$\wedge$(30.0, 30.5)$\wedge$(0.0, 0.0)] & $\wedge$(31.0, 31.5)$\wedge$(30.1, 30.2)$\wedge$(0.0, 0.1)] & & & &\\
\hline
\end{tabular}

\end{table*}

\section{Discussion and Future Work}

In this work, we have proposed a new state space exploration technique \nexg that is an improvement over existing state space approaches. 
In addition to out-performing state of the art falsification techniques, our technique enables the control designer to develop custom algorithms for state space exploration and generate trajectories that navigate the state space along with additional constraints.

We also would suggest two additional use-cases of our approach: generating near-miss trajectory and estimated trajectories. Instead of continuing to run the \reachDest algorithm to search for execution that reaches the target state, a control designer can terminate early and use a custom state space exploration using inverse sensitivity to generate \emph{near miss} safety instances where the trajectory approaches the set of unsafe states within a threshold.
\begin{figure}[!ht]
\centering     
\includegraphics[width=60mm]{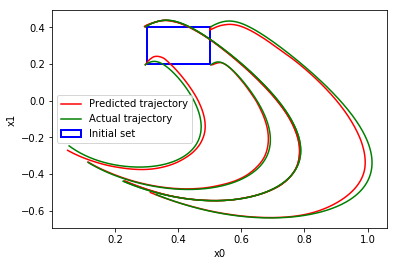}
\caption{\small{Predicting system trajectories using sensitivity approximation. We use the sensitivity approximator to perturb an initial anchor trajectory starting at the centroid of the rectangle to predict the four trajectories starting from its corners.}}
\label{fig:fwd-sen-bench4}
\end{figure}
Secondly, if invoking the exact simulation engine becomes computationally expensive (e.g. exploring a high-dimensional space), then the control designer can generate approximate trajectories from a given initial state by using a sensitivity estimator.
An example of generating such approximate trajectories is given in Figure~\ref{fig:fwd-sen-bench4}.

\textbf{Future Work:} As our choice of network architecture, training data, scaling factor and other parameters is mostly empirical, we plan to investigate different architectures to accelerate network training and explore other parameters configurations to improve performance of \nexg. 
%
%
We aim to extend this work to handle more generic systems such as feedback systems with environmental inputs. 
%
%
Further, we would like to enhance our falsification technique to the general class of STL specification and automate its integration with S-TaLiRo. We finally aim to explore ways to integrate our method into frameworks used for generating adversarial executions during control synthesis.

\bibliographystyle{plain} 
\bibliography{nexgArxiv} 
 

\end{document}